\documentclass[final,onefignum,onetabnum]{siamonline190516}

\usepackage{lipsum}
\usepackage{amsfonts}
\usepackage{graphicx}
\usepackage{epstopdf}
\usepackage{algorithmic}
\ifpdf
  \DeclareGraphicsExtensions{.eps,.pdf,.png,.jpg}
\else
  \DeclareGraphicsExtensions{.eps}
\fi

\usepackage{enumitem}
\setlist[enumerate]{leftmargin=.5in}
\setlist[itemize]{leftmargin=.5in}

\newsiamremark{remark}{Remark}
\newsiamremark{hypothesis}{Hypothesis}
\crefname{hypothesis}{Hypothesis}{Hypotheses}
\newsiamthm{claim}{Claim}

\headers{Inverse Gaussian Process regression for likelihood-free inference}{H. Wang, Z. Ao, T. Yu and J. Li}

\title{Inverse Gaussian Process regression for likelihood-free inference\thanks{\today.
\funding{This work was partially 
supported by the NSFC under grant number 11301337.}}}

\author{Hongqiao Wang\thanks{School of Mathematics and Statistics, Central South University, Changsha, China.}
\and Ziqiao Ao\thanks{School of Mathematics, University of Birmingham, Edgbaston, Birmingham B15 2TT, UK.}
\and Tengchao Yu \thanks {School of Mathematical Sciences,  
Shanghai Jiao Tong University, 800 Dongchuan Rd, Shanghai 200240, China.}
\and Jinglai Li\thanks{Corresponding author, School of Mathematics, University of Birmingham, Edgbaston, Birmingham B15 2TT, UK.
Email: \email{j.li.10@bham.ac.uk}.}}

\usepackage{amsopn}

\ifpdf
\hypersetup{
  pdftitle={Inverse Gaussian Process regression for likelihood-free inference},
  pdfauthor={H. Wang, Z. Ao,  T. Yu and J. Li}
}
\fi

\usepackage{bm,amsmath}

\DeclareMathAlphabet\mathpzc{OT1}{pzc}{m}{it}
\let\mathcal=\mathpzc

\def\gp{{\mathrm{GP}}}

\numberwithin{equation}{section}

\def\intinfty{\int\limits_{\!\!-\infty\,\,}^{\,\,\infty\!\!}\kern-0.0em}
\def\iintinfty{\mathop{\int\!\!\int}\limits_{\!\!-\infty\,\,}^{\,\,\infty\!\!}\kern-0.0em}
\def\iiintinfty{\mathop{\int\!\!\int\!\!\int}\limits_{\!\!-\infty\,\,}^{\,\,\infty\!\!}\kern-0.0em}

\let\<=\langle
\let\>=\rangle
\let\^=\hat
\def\~#1{{\-ox{\sf#1}}}

\def\btheta{{\bm{\theta}}}

\let\-=\mathbf

\newtheorem{assumption}[theorem]{Assumption}

\def\@#1{{\cal #1}}

\begin{document}

\maketitle

\begin{abstract}
  In this work we consider Bayesian inference problems with intractable likelihood functions.
We present a method to compute an approximate of the posterior with a limited number of model simulations. 
 The method features an inverse Gaussian Process regression (IGPR), i.e., one from the output of a simulation model to the input of it.  
Within the method, we provide an adaptive algorithm with a tempering procedure to construct the approximations of the marginal posterior distributions. 
With examples we demonstrate that IGPR has a competitive performance compared to some commonly used algorithms, 
especially in terms of statistical stability and computational efficiency,
while the price to pay is that it can only compute a weighted Gaussian approximation of the marginal posteriors. 
\end{abstract}


\section{Introduction}
\label{sec:intro}
Bayesian inference~\cite{gelman2014bayesian} provides a natural framework 
for estimating parameters of interest (PoI) from indirect, incomplete and noisy observations. 
A major advantage of the Bayesian framework is that the posterior distribution 
is a probabilistic representation of the PoI, 
and thus can characterize the uncertainty information in the inference results. 
Computing the posterior distribution, i.e. the distribution of the PoI conditional on the observations, is therefore a central task of Bayesian inference. 
In many practical problems, often the likelihood function is intractable (namely, evaluating the likelihood function 
is either computationally prohibitive or simply impossible)~\cite{sisson2018handbook},
rendering most commonly used methods, such the Markov chain Monte Carlo (MCMC) simulation~\cite{andrieu2003introduction}
and the variational Bayes method~\cite{fox2012tutorial,blei2017variational} infeasible. 
Likelihood-free inference methods, which do not require a tractable form of the likelihood function,  are needed  to address such problems.  

Considerable efforts have been devoted to the development of  approximate inference methods and 
in what follows we briefly overview three classes of approximate inference methods, 
which use different strategies to avoid the likelihood evaluation. 
First the approximate Bayesian computations (ABC)~\cite{beaumont2002approximate,sisson2018handbook} are 
probably the most popular class of approximate inference methods. 
The main idea behind ABC is to generate samples from a prescribed scheme (e.g. from the prior distribution)
and then determine whether a sample of the POI can be accepted
as a posterior sample based on the discrepancy between the data simulated from the sample and the observed data.
The ABC methods have been successfully applied to a wide range of practical problems,
but they typically require a large number of simulations of the underlying mathematical model, 
which makes them highly costly for problems with computationally intensive mathematical models. 
Another often used strategy is to construct computationally inexpensive surrogate or approximate model of the likelihood or a related function from the simulated data,
and then use the resulting surrogate models in a standard  posterior computation such as MCMC. 
A typical example of such methods is the synthetic likelihood~\cite{wood2010} that approximates
the likelihood function from a normal density estimate for the summary statistics and others such as \cite{thomas2020likelihood,papamakarios2019sequential}.
Finally in this work our interest lies in a type of  methods that aim to directly approximate the posterior distribution from the simulated samples. 
Specifically they assume a certain conditional distribution density model for the posterior, 
and then train the conditional density model through the simulated samples. 
In what follows we shall refer to such approaches as the posterior approximation methods. 
Works along this line include~\cite{lueckmann2017flexible,papamakarios2016fast}, just to name a few.
 Note that  in the conditional density models the input is the observed data and the output of it is the distribution of POI. 
Since in this type of methods the conditional density models are often constructed with complex models such as artificial neural networks, 
training such models are often computationally intensive. 

In many practical problems often the users may only be interested in the marginal posteriors of each parameters, rather than the joint 
distribution as a whole.
The goal of this work is to propose a  method to compute the marginal posterior distributions which can be efficiently and stably implemented.  
The basic idea of the proposed method is to use Gaussian Process (GP) regression~\cite{sacks1989design,kennedy2001bayesian,williams2006gaussian} as the conditional density model for approximating the actual posterior. 
Unlike many usual regression methods that provide a point estimate of the responses, 
the GP method computes a Gaussian distribution of the response conditional on the predictors,
which provides a natural framework to approximate  the posterior. 
It should be noted that the GP methods have been previously used to accelerate the posterior computation in several 
works~\cite{kandasamy2015bayesian,wang2018adaptive,acerbi2018variational}. 
In these works,  GP is mainly used as a regression tool, constructed to approximate the likelihood function or the log-likelihood function:
 namely the predictor is the PoI and the response is the value of the (log)-likelihood. As a result these methods can not handle problems where the likelihood function is intractable without extra measures. 
In the present work we use the GP model to directly approximate the marginal posterior distribution,  which is a mapping from observation data to 
the distribution of PoI.
As will be discussed later, in the present setup, the data and the POI are respectively the output and the input of a simulation model,
and so our method essentially constructs a GP regression from the output of a model to the input of it. 
For this output-to-input structure we refer to the proposed method as Inverse Gaussian Process Regression~(IGPR). 

In all the posterior approximation methods including IGPR, one of the most important issues is to construct the training set~\cite{williams2006gaussian}, i.e. to determine locations where to perform the simulation of the mathematical model.
Since we consider problems with expensive simulation models, determine the sampling points are essential 
for the computational efficiency. 
This task, however, is particularly challenging in the present problems,
 as the inputs of the condition distribution model (namely, the GP model in our method) are actually the output of a simulation model, which makes it impossible to
 draw a sample from an exact location even if such a location is specified. To address the issue, we design an adaptive algorithm to generate sample points and update the resulting GP approximation of the posterior. 
Since the training process of GP is less costly especially as we use the local GP model that will be explained later, 
the IGPR method is usually more efficient to implement than those based on artificial neural networks,
while the price we pay  is that IGPR can only approximate the marginal posteriors. 
Theoretically we are able to show that the IGPR method can correctly recover the mean  of the marginal posterior distribution in 
the large sample limit. 
Finally our numerical experiments demonstrate that, IGPR has competitive performance in terms of accuracy,
 it is more statistically robust, in the sense that its results are subject to less statistical fluctuations,
 and finally it is more computationally efficient than the methods based on artificial neural networks.

The rest of the paper is organized as follows. In Section~\ref{sec:method} we discuss the 
setup of the approximate inference problems, and provide some preliminaries that will be used in the proposed method.
In Section~\ref{sec:igpr} we describe the proposed IGPR method in details and in Section~\ref{sec:examples} we provide some examples to demonstrate its performance. Finally we offer some closing remarks in Section~\ref{sec:conclusions}.

\section{Problem setup and preliminaries}
\label{sec:method}
In this section, we discuss some important preliminaries of the proposed method,
and we start with the basic setup of the inference problems that we want to solve.

\subsection{The likelihood-free Bayesian inference problem}
Let us consider a Bayesian inference problem: let $\btheta\in  R^{n_\theta}$ be the parameters of interest, 
and $\-d\in R^{n_d}$ be the observed data, 
and we want to compute the posterior distribution:  
\begin{equation}
\pi(\btheta|\-d) =\frac{\pi(\-d|\btheta)\pi(\btheta)}{\pi(\-d)},\label{e:posterior}
\end{equation}
where $\pi(\-d|\btheta)$ is the likelihood function and $\pi(\btheta)$ is the prior distribution on $\btheta$.
In certain problems, the likelihood function can be directly evaluated, and in this case 
the posterior samples can be drawn using the MCMC simulation. 
In this work, however, we assume that evaluation of the likelihood function $\pi(\-d|\btheta)$ is infeasible,
while it is possible to generate data from the likelihood function for a given value of $\btheta$. 
A typical setup for such problems is that the data is a function of both the parameter $\btheta$ and certain random variables $\bm{\xi}$:
\begin{equation}
\-d=G(\btheta,\bm{\xi}),
\end{equation}
where the distribution of $\bm\xi$ is known.
In some situations $\bm\xi$ can be a stochastic process or a random field. 
It should be clear that when the model $G(\cdot,\cdot)$ is complex, for example described by a nonlinear differential equation system, it is extremely difficult to 
derive the associated likelihood function. 
This type of problems appear quite often in the field of computational biology ranging from population dynamics to biochemical networks~\cite{lillacci2010parameter,chou2009recent}. 
On the other hand, in such problems, it is required that for a given parameter $\hat{\btheta}$, synthetic data  can be generated 
by performing a simulation of the underlying mathematical model $G(\hat{\btheta},\bm{\xi})$. 
To avoid confusion, in what follows we refer to the synthetic data as \textit{samples}, denoted as $\hat{\-d}$, and the observed data as 
\textit{data}, denoted as $\tilde{\-d}$.

\subsection{The approximate Bayesian Computation}
Due to the intractability of the likelihood function, it is not possible to compute the posterior distribution~\eqref{e:posterior} 
exactly. Instead, we seek to find an approximate posterior distribution for the problem. 
As is discussed in Section~\ref{sec:intro}, one popular approach to do this is the ABC method.
The basic idea of ABC is to approximate the intractable likelihood function by
\begin{equation}
\pi_\epsilon(\tilde{\-d}|\btheta) = \int \delta_\epsilon(\rho(\tilde{\-d},\hat{\-d})) \pi(\hat{\-d}|\btheta) d\-d
\end{equation}
where $\delta_\epsilon(\cdot)$ is an indicator function: $\delta_\epsilon(z)=1$ if $z\leq\epsilon$ and $\delta_\epsilon(z)=0$ otherwise,
 $\rho(\cdot,\cdot)$ is a distance measure of two data points, and $\epsilon>0$ is a prescribed threshold.
 Throughout this paper $\rho$ is taken to be the Euclidian distance unless otherwisely stated. 
It should be clear that $\pi_\epsilon(\tilde{\-d}|\btheta)=\pi(\tilde{\-d}|\btheta)$ for $\epsilon=0$.
The posterior distribution $\pi(\btheta|\tilde{\-d})$ can thus be approximated by
\begin{equation}
\pi_\epsilon(\btheta|\tilde{\-d}) = \frac1{\pi_\epsilon(\tilde{\-d})}\int \delta_\epsilon(\rho_{abc}(\tilde{\-d},\hat{\-d})) \pi(\hat{\-d}|\hat{\btheta})\pi(\btheta) d\-d,
\end{equation}
where $\pi_\epsilon(\tilde{\-d})$ is the normalization constant. 
The very basic ABC rejection algorithm generates a set of approximate posterior samples by repeatedly performing the following steps, 
\begin{enumerate}
\item draw $\hat{\btheta}$ from $\pi(\btheta)$;
\item  generate $\hat{\bm{d}}\sim \pi(\cdot|\hat{\btheta})$;
\item if $\rho(\tilde{\bm{d}},\hat{\bm{d}})<\epsilon$, accept $\hat{\btheta}$.
\end{enumerate}
We reinstate that the procedure given above is the very basic version of ABC with rejection (ABC-REJ), 
and more sophisticated ABC algorithms such as the ABC with MCMC~\cite{marjoram2003markov} and the ABC with sequential Monte Carlo (ABC-SMC)~\cite{sisson2007sequential} are also available.
We refer to the review articles~\cite{beaumont2002approximate,sisson2018handbook} as well as the references therein for more information on the ABC algorithms. 


\subsection{The GP regression with noise} 
\label{sec:GPR}
We now provide a very brief introduction to  the GP regression, which is the main tool of the proposed method, 
and we want to note here that the GP model is slightly modified for the use in approximate inference.  
Simply speaking the GP regression performs a nonparametric regression in a Bayesian framework~\cite{williams2006gaussian}.
Specifically, given $m$ data pairs $D=\{(\-x^*_i,y^*_i)\}_{i=1}^m$
with $\-x\in R^{n_x}$ and $y\in R$ being the explanatory and response variables respectively,
the task is to predict the value of $y$ at a new point ${\-x}$,
and more precisely we want to compute the conditional distribution $\pi(y|{\-x},D)$.  
A standard model for regression is to assume that the explanatory and response variables are related via an underlying function: 
\begin{equation}
y = f(\-x)+{\zeta},
\end{equation}
where $\zeta$ is an observation noise following $N(0,\nu)$. 
The main idea of the GP method is to assume that the underlying function $f(\-x)$ is a Gaussian random field defined on $R^{n_x}$,
whose mean is $\mu(\-x)$ and covariance is specified by
a kernel function $k(\-x,\-x')$, 
 namely,
\[ \mathrm{cov}[f(\-x),f(\-x')] = k(\-x,\-x'). \]
As a result the joint distribution of $(\-Y,\,f({\-x}))$ is, 
\begin{equation}  
\left[ \begin{array}{c}
         \-Y \\
        f({\-x}) \end{array} \right] \sim \@N\left(\begin{array}{c}
         \mu(\-X) \\
        \mu({\-x}) \end{array},
				\left[
				\begin{array}{ll}
         K(\-X,\-X)+\nu I &K(\-X,{\-x}) \\
        K({\-x},\-X) &K({\-x},{\-x}) \end{array}\right]
				\right) ,  \label{e:jointdis1}
\end{equation}
where  $\-Y = \left[ {y}^*_1, \ldots, {y}^*_m\right]$, $\-X = \left[\-x^*_1, \ldots, \-x^*_m\right]$,
$I$ is the identity matrix, 
and the notation $K(\-A,\-B)$ denotes the matrix of the covariance evaluated at all pairs of points in set $\-A$ and in set $\-B$.
We note here that Eq.~\eqref{e:jointdis1} is the usual form of the GP model with noisy observations~\cite{williams2006gaussian},
where the noise-free prediction $f({\-x})$ is sought. 
However, for our purposes, we are actually interested in the noisy prediction to account for the intrinsic uncertainty in the underlying model itself (i.e.,
even if $f(x)$ is exactly know, $y$ is subject to uncertainty due to the presence of noise),
and the matter will be discussed further in Section~\ref{sec:igpr}. Here we can derive the joint distribution of $(\-Y, y)$:
\begin{equation}  
\left[ \begin{array}{c}
         \-Y \\
       y \end{array} \right] \sim \@N\left(\begin{array}{c}
         \mu(\-X) \\
        \mu({\-x}) \end{array},
				\left[
				\begin{array}{ll}
         K(\-X,\-X)+\nu I &K(\-X,{\-x}) \\
        K({\-x},\-X) &K({\-x},{\-x}) +\nu\end{array}\right]
				\right) .  \label{e:jointdis2}
\end{equation}

It follows immediately that the conditional distribution $\pi(y|\-x,D)$ is also Gaussian: 
\begin{subequations}
\label{e:gp}
\begin{equation}
  y ~|~\-x,~D \sim\mathcal{N}(\mu_\gp(\-x), \sigma_{\mathrm{GP}}^2(\-x)), \label{e:post}
\end{equation}
and the posterior mean $\mu_\gp$ and variance $\sigma_\gp^2$ are available in explicit forms:
\begin{align}
\label{eq:mu}
&\mu_\gp(\-x)=\mu(\-x)+k(\-x,\-X)(k(\-X,\-X)+\nu I)^{-1}(\-Y-\mu(\-x)\-1_{m}),\\
\label{eq:var}
&\sigma_\gp^2(\-x)= k(\-x,\-x)-k(\-x,\-X)(k(\-X,\-X)+\nu I)^{-1}k(\-X,\-x)+\nu,
\end{align}
\end{subequations} 
where $\-1_{m}$ is a $m$-dimensional vector of ones.
In what follows we use the notation $\pi_\gp(y|\-x,D)$ to denote the 
Gaussian distribution given by the GP model in Eq.~\eqref{e:gp}. 
It is important to note here that the noise covariance $\nu$ is usually not known in advance,
and in practice, it along with other hyper-parameters is determined 
by the maximum likelihood estimation (MLE).
 We refer to \cite{williams2006gaussian} for more details.  

\section{The IGPR method}\label{sec:igpr}
We now discuss how to compute the approximate posterior distribution using the GP model,
and in particular we want to reinstate here that our goal is to compute the posterior marginal distribution of each parameter. 
Specifically, starting with a basic version of IGPR model in Section~\ref{sec:igpr1}, and making use of a formulation based on proposal prior
as is shown in Section~\ref{sec:igpr2}, 
we develop an adaptive IGPR algorithm in Section~\ref{sec:igpr3}.

\subsection{A basic IGPR model}\label{sec:igpr1}
In this section we present the basic version of the IGPR method. 
 We want to compute the marginal posterior distributions $p(\theta^j|\-d)$ where $\theta^j$ 
is the $j$-th component of $\btheta$ for any $0<j\leq n_\theta$. 
Now suppose that we have a set of sample points generated from the likelihood function 
$\{(\hat{\btheta}_i,\hat{\-d_i})\}_{i=1}^m$.
Define the training set 
as $D^j= \{...,(\hat{\-d}_i,\hat{\theta}_i^j),...\}$, and by assuming $\theta^j(\-d)$ is Gaussian random field, 
 we can construct a GP model $\pi_{\gp}(\theta^j|\-d,D^j)$ from  $D^j$.
Plugging the observed data $\tilde{\-d}$ into the GP model yields the GP based posterior distribution 
$\pi_{\gp}(\theta^j|\tilde{\-d},D^j)$, which 
 can be used as an approximation of the true posterior $\pi(\theta^j|\tilde{\-d})$. 
 We refer to the GP-based approximate posterior as the IGPR model.
 In particular we note that the variance in the IGPR model consists of two parts: one corresponds to the actually variance of the posterior, and 
 the other corresponds to the model uncertainty of GP due to the limited number of samples.
  To obtain the IGPR model, we need to construct a training set $D^j$, and a natural way of doing so 
is to draw the training sample points from $\pi(\btheta,\-d)$: namely we
first draw samples of the parameter of $\btheta$ from the prior distribution $\pi(\btheta)$, obtaining $\hat{\btheta}_1,...,\hat{\btheta}_m$,  
 and for each sample $\hat{\btheta}_i$, we then generate a $\hat{\-d}_i$ from the likelihood function $\pi(\cdot|\hat{\btheta}_i)$.
 Next following the idea of ABC, we introduce a cut-off distance $\epsilon$, 
 and only include the samples whose distance to the observed data point $\tilde{\-d}$ is smaller than $\epsilon$ in the training set, yielding,
 $$D_\epsilon^j= \{(\hat{\bm{d}}_i,\hat{\theta}_i^j)|\rho(\hat{\bm d}_i,\tilde{\bm d})<\epsilon\}_{i=1}^m.$$
 This step can also be understood as employing local GP construction~~\cite{gramacy2015local,wu2016surrogate} and the motivation is two-fold.
 First, it eliminates the influence of samples that are far away from the actual data point $\tilde{\-d}$, which can potentially improve the accuracy of the GP prediction at $\tilde{\-d}$.
Second, for problems where the data set is massively large, there may be some computational issues with the GP model~(e.g., the inversion of the large covariance matrix $(k(\-X,\-X)+\sigma_n^2I)$ may not be numerically stable) and the use of local GP alleviates these problems. 
  Using the resulting training set $D_\epsilon^j$ and the GP model described in Section~\ref{sec:GPR} , we are able to compute a Gaussian approximation of each marginal distribution,  $\pi_{GP}(\theta^j|\tilde{\bm{d}},D_\epsilon^j)$ and we define
  \begin{equation}
  \psi_{\gp}(\btheta) =\prod_{j=1}^{n_\theta}\pi_\gp(\theta^j|\tilde{\-d},D_\epsilon^j),
  \end{equation}
  which can be regarded as an approximation of the joint posterior. 
 Alg.~\ref{alg:igpr1} provides the pseudo code of this procedure. 
  \begin{algorithm}
 \caption{The basic IGPR algorithm}\label{alg:igpr1}
 \begin{itemize}
 \item draw $\hat{\btheta}_1,...,\hat{\btheta}_m$ from $\pi(\btheta)$;
 \item for $i=1...m$, draw $\hat{\bm{d}}_i\sim \pi(\bm{d}|\hat{\btheta}_i)$;
 \item for $j=1$ to $n_\theta$ do
 \begin{itemize}
 \item let $D_\epsilon^j= \{(\hat{\bm{d}}_i,\hat{\theta}_i^j)|\rho(\hat{\bm d}_i,\tilde{\bm d})<\epsilon\}_{i=1}^m$;
 \item construct the GP model $\pi_\gp(\theta^j|\bm{d},D_\epsilon^j)$ from data set $D_\epsilon^j$;
 \end{itemize}
 \item let   $\psi(\btheta) =\prod_{j=1}^{n_\theta} \pi_\gp(\theta^j|\tilde{\bm{d}},D_\epsilon^j)$.
 \end{itemize}
\end{algorithm}

The algorithm presented above is  an analogy of ABC-REJ in the ABC family, 
in that both of them directly draw samples from the prior distribution and then select samples based on their distance to the observed data point. 
Compared to the two more complicated versions of IGPR to be discussed later,  an advantage of this version of IGPR model is that it can be implemented 
in an on-line/off-line fashion: in the off-line stage a very large number of samples can be drawn from the model and form  the data set 
for constructing the GP model;
in the on-line stage, once an observed data becomes available, it can be directly fed into the GP model and yield the approximate posterior. 
To this end, though this version of IGPR is relatively simple, it can be useful to applications 
 where the online/offline scheme is suitable.

\subsection{IGPR with a proposal prior}\label{sec:igpr2}
In Alg.~\ref{alg:igpr1}, the samples of $\btheta$ are drawn from the prior distribution $\pi(\btheta)$. 
Since in many practical problems the procedure of generating $\hat{\-d}_i$ from the likelihood function $\pi(\cdot|\hat{\btheta}_i)$ can be highly expensive, 
simply sampling according to the prior distribution
may not be a good choice for the computational efficiency, especially when the posterior differs significantly from the prior. 
Intuitively a more efficient way to determine the training samples is
to draw samples from a distribution that is close to the posterior, which is discussed in  this section.  

Suppose that we are given a distribution that can approximate the posterior, and following ~\cite{papamakarios2016fast} it is referred to as the proposal prior $\pi_q(\btheta)$. 
In this section we shall discuss how to construct the IGPR model using such a proposal prior. 
First it is easy to see that the posterior distribution can be written as, 
\begin{equation}
\pi(\btheta|\-d) \propto {\pi(\-d|\btheta)\pi(\btheta)} = {\pi(\-d|\btheta)}\pi_q(\btheta)\frac{\pi(\btheta)}{\pi_q(\btheta)}
\propto \pi_q(\btheta|\-d)\frac{\pi(\btheta)}{\pi_q(\btheta)} \label{e:bayes_prop}
\end{equation}
where 
\begin{equation}\pi_q(\btheta|\-d)=\frac{\pi(\-d|\btheta)\pi_q(\btheta)}{\pi_q(\-d)},\quad
\pi_q(\-d)=\int {\pi(\-d|\btheta)}\pi_q(\btheta) d\btheta. \label{e:post_q}
\end{equation}

Next we can apply  Algorithm~\ref{alg:igpr1} to $\pi_q(\btheta|\tilde{\-d})$ in Eq.~\eqref{e:post_q} (except that the samples 
are now drawn from $\pi_q$ instead of $\pi$), and obtain the Gaussian approximation $\psi(\btheta)$.
It should be clear that the Gaussian distribution $\psi(\btheta)$ computed is an approximation of $\pi_q(\btheta|\-d)$
 rather than  the actual posterior distribution as it uses 
$\pi_q$ as the prior.   Now replacing $\pi_q(\btheta|\-d)$ in Eq.~\eqref{e:bayes_prop} yields an approximation to the actual posterior $\pi(\btheta|\tilde{\-d})$,
denoted by, 
\begin{equation}
\pi_{post}(\btheta)\propto\psi(\btheta)\frac{\pi(\btheta)}{\pi_q(\btheta)}.
\end{equation}
Next we construct an independent Gaussian approximation of the prior,  
\begin{equation}
 \phi_{0}(\btheta)= \prod_{j=1}^{n_\theta} \mathcal{N}(\mu^j_0,{(\sigma^j_0)}^2),
 \end{equation}
where $(\mu^j_0,{(\sigma^j_0})^2)$ are respectively the prior mean and variance of $\theta^j$ for $j=1...n_\theta$. 
It should be noted there that we can use such an independent Gaussian approximation of the prior because we only intend to 
compute the marginal posterior distributions. 
If we take  $\pi_q(\btheta)$ to be an independent Gaussian, 
\begin{align*}
    \pi_q(\btheta)&\sim \prod_{j=1}^{n_\theta} \mathcal{N}(\mu^j_q, {(\sigma^j_q)}^2),
\end{align*}
we can derive that $\phi(\btheta)\propto\psi(\btheta)\phi_0(\btheta)$ is also in the form of:
\begin{subequations}\label{e:phi}
\begin{equation}
\phi(\btheta) = \prod_{j=1}^{n_\theta} \mathcal{N}(\mu^j_{\phi}, ({\sigma^j_{\phi}})^2),
\end{equation}
 with mean
\begin{equation}
\mu_{\phi}^j = \frac{\frac{\mu_{\gp}}{(\sigma^j_{\gp})^2} - \frac{\mu_q}{{(\sigma_q^j)}^2} + \frac{\mu_0}{{(\sigma_0^j)}^2}}{\frac{1}{(\sigma^j_{\gp})^2} - \frac{1}{{(\sigma_q^j)}^2} + \frac{1}{{(\sigma_{0}^j)}^2}},
\end{equation}
and variance
\begin{equation}
{\sigma^j}^2_{\phi} = \frac{1}{\frac{1}{(\sigma^j_{\gp})^2} - \frac{1}{({\sigma_q^j})^2} + \frac{1}{({\sigma_{0}^j})^2}}.
\end{equation}
\end{subequations}
It is easy to see that, for $\phi$ to be a well-defined Gaussian distribution, we must have
 \begin{equation}
\frac{1}{(\sigma^j_{\gp})^2} - \frac{1}{({\sigma_q^j})^2} + \frac{1}{({\sigma_{0}^j})^2}>0,\label{e:ineq_var}
 \end{equation}
and  measures can be taken to ensure this in the numerical implementation:
 e.g., one can choose $\pi_q$ such that  $\sigma_q^j\leq {\sigma_0^j}$ for all $j=1...n_\theta$. 
 Next it follows that the approximate posterior $\pi_{post}$ can be written as
  $$\pi_{post}\propto\phi(\btheta)\frac{\pi(\btheta)}{\phi_0(\btheta)}.$$
 We note here that in Algorithm~\ref{alg:igpr2}, the third step inside the ``for''-loop is employed  to make sure that the posterior variance is smaller than
 the prior variance which in turn ensures that Eq.~\eqref{e:ineq_var} holds. 
 
   \begin{algorithm}
 \caption{The proposal prior based IGPR algorithm}\label{alg:igpr2}
 \begin{itemize}
 \item draw $\hat{\btheta}_1,...,\hat{\btheta}_m$ from $\pi_q(\btheta)$;
 \item for each $\hat{\btheta}_i$ for $i=1...m$, draw $\hat{\bm{d}}_i\sim \pi(\bm{d}|\hat{\btheta}_i)$;
 \item for $j=1$ to $n_\theta$ do
 \begin{itemize}
 \item let $D^j= \{(\hat{\bm{d}}_i,\hat{\theta}_i^j)|\rho(\hat{\bm d}_i,\tilde{\bm d})<\epsilon\}_{i=1}^m$;
 \item construct the GP model $\pi_\gp(\theta^j|\bm{d},D^j)$ from data set $D^j$;
 \item if $\sigma^j_{\gp}>\sigma_0^j$, let $\sigma^j_{\gp}=\sigma^j_0$;
 \end{itemize}
 \item let $\psi(\btheta) =\prod_{j=1}^{n_\theta} \pi_\gp(\theta^j|\tilde{\bm d},D^j)$;
 \item computer $\phi(\btheta)$ using Eqs.~\eqref{e:phi};
 \item let $\pi_{post}(\btheta)\propto\phi(\btheta)\frac{\pi(\btheta)}{\phi_0(\btheta)}.$
 \end{itemize}
\end{algorithm}
In the proposal prior based method, the key is to obtain a good proposal $\pi_q$. As has been discussed, the samples used to construct the GP approximation 
are drawn from the proposal prior $\pi_q$, and intuitively, for these samples to be informative, we need them to cover the high probability region of the posterior.
As such a natural choice is to select $\pi_q$ to be close to the posterior distribution. 
Certainly this cannot be done in one step, as the posterior is not known in advance, and so in next section we introduce a scheme to adaptively identify the proposal prior and construct the IGPR model for a given posterior distribution.

\subsection{An adaptive IGPR scheme} \label{sec:igpr3}
Here we discuss how to construct the proposal prior $\pi_q(\btheta)$.
As is discussed in the previous section, computing a good proposal prior $\pi_q$ is rather challenging especially when the posterior is concentrated 
in a rather small region of the entire prior state space. 
One remedy to the issue is to use a simulated tempering approach~\cite{brooks2011handbook}. Namely we construct a sequence of ``intermediate posterior distributions" with 
the first one  being the prior and the last one being the actual posterior: $\{\pi_t(\btheta|\tilde{\-d})\}_{t=0}^T$ with $\pi_0(\btheta|\tilde{\-d})=\pi(\btheta)$ and $ \pi_T(\btheta|\tilde{\-d})=\pi(\btheta|\tilde{\-d})$.
Typically, the sequence of ``intermediate posteriors'' should be constructed in a way that their variances are decreasing. 
One can see that the standard tempering/annealing approaches do not apply here as we do not have an explicit expression of the posterior distribution.
Therefore, we use a special tempering strategy by creating artificial data with larger observation noise. 
Namely we introduce a sequence of artificial  data noise:
$\eta_t\sim\@N(0,\sigma^2_tI).$  
Here the standard deviation $\sigma_t$ plays the role of the tempering parameter, and it should be chosen to be gradually decreasing with respect to $t$ 
and $\sigma_T=0$.
We then define for any $1\leq t\leq T$, 
$$\tilde{\-d}'_t=\tilde{\-d}+\bm\eta_t,\quad \mathrm{and}\quad \pi_t(\btheta|\tilde{\-d}) =\pi(\btheta|\tilde{\-d}'_t), $$ and 
one can see that $\pi_T(\btheta|\tilde{\-d}) =\pi(\btheta|\tilde{\-d}'_T)=\pi(\btheta|\tilde{\-d})$.
Other tempering approaches can also be used to design the posterior sequence as long as they satisfy the aforementioned conditions. 

Once a tempering scheme is chosen, the adaptive scheme proceeds as follows. 
For any $1\leq t\leq T$, let $\phi_{t-1}(\btheta)$ be the independent Gaussian approximation of the posterior of the previous step $\pi_{t-1}(\btheta|\tilde{\-d})$, 
and we then use $\phi_{t-1}(\btheta)$ as the proposal prior for the present step to compute $\phi_t(\btheta)$, the independent Gaussian approximation of 
$\pi_{t}(\btheta|\tilde{\-d})$.  As has been explained above $\pi_{t}(\btheta|\tilde{\-d})$ is 
essentially $\pi(\btheta|\tilde{\-d}'_t)$, and thus they key step here is to construct IGPR model for $\pi(\btheta|\tilde{\-d}'_t)$. 
Specifically, we first generate the simulated data points $\{\hat{\-d}_i,\btheta_i)\}_{i=1}^m$ from $\phi_{t-1}(\btheta)\pi(\-d|\btheta)$;
for each $i=1...m$, we let $$\hat{\-d}'_i = \hat{\-d}_i+\bm{\eta}_t,\quad\mathrm{with} \quad \bm\eta_t\sim \@N(0,\sigma_t^2I);$$
the last step is to construct the IGPR model using the data set $\{\hat{\-d}'_i,\btheta_i)\}_{i=1}^m$ as described earlier. 
The procedure is repeated until it reaches the last step $t=T$. 
It is important to note here that, in Algorithms~\ref{alg:igpr1} and \ref{alg:igpr2}, it is assumed that the cut-off distance $\epsilon$ is prescribed, and here we allow that $\epsilon_t$ in each step is determined by choosing a certain portion $\omega$ of the samples in terms of their distances 
to the data point $\tilde{\-d}$. By determining the cut-off distance this way  we can ensure that sufficient samples are used to construct the GP model. 
We conclude this section by presenting the complete procudure of the adaptive IGPR  in Alg.~\ref{alg:igpr}.

\begin{algorithm}
    \caption{The adaptive IGPR algorithm}
    \label{alg:igpr}
    \begin{algorithmic}[1]
    
    \STATE{\textbf{Algorithm Parameters}:  $m$, $T$, $\omega$,  $\{\sigma_t\}_{t=1}^T$;}
     \STATE{\textbf{Output}: ${\pi}_{post}(\btheta|\bm{\tilde{d}})$;}
       \STATE let 
$ \phi_{0}(\btheta)= \prod_{j=1}^{n_\theta} \mathcal{N}(\mu^j_0,(\sigma^j_0)^2)$,
where $(\mu^j_0,{(\sigma^j_0})^2)$ are respectively the prior mean and variance of $x^j$;

  \FOR {$t = 1$ to $T$}
         \STATE draw $m$  samples $\{\hat{\bm{\theta}}_1,...,
				\hat{\bm{\theta}}_{m}\}$ from the proposal prior $ \phi_{t-1}(\btheta)$; 
\FOR {$i=1$ to $m$}				
       \STATE draw $\hat{\bm{d}}_i \sim \pi(\cdot|\hat{\bm{\theta}}_i)$;
  \STATE let
          $\hat{\bm{d}}'_i = \hat{\bm{d}}_i+\bm{\eta}_t$ with $\bm\eta_t\sim \@N(0,\sigma_t^2I)$;
          \STATE calculate $\delta_i=\|\hat{\bm{d}}'_i -\tilde{\bm d}\|_2$;
   \ENDFOR
   \STATE let $\epsilon_t$ be the $\omega$-th quantile of $\{\delta_i\}_{i=1}^m$;
	\FOR {$j=1$ $n_\theta$}
	\STATE let $D^j= \{(\hat{\bm{d}}'_i,\hat{\theta}_i^j)|\delta_i<\epsilon_t\}_{i=1}^m$;
 \STATE construct the GP model $\pi_\gp(\theta^j|\bm{d},D^j)$ from data set $D^j$;
\IF{$\sigma^j_{\gp}>\sigma_0^j$} \STATE let $\sigma^j_{\gp}=\sigma^j_0$;\ENDIF

	\ENDFOR
	\STATE   let $\psi_t(\btheta) =\prod_{j=1}^{n_\theta} \pi_\gp(\theta^j|\tilde{\bm{d}},D^j)$;

\STATE compute $\phi_t(\btheta)\propto \psi_t(\btheta)\phi_0(\btheta)$ using Eqs.~\eqref{e:phi};

	\ENDFOR
\STATE	 let $\pi_{post}(\btheta)\propto\phi_T(\btheta)\frac{\pi(\btheta)}{\phi_0(\btheta)}.$
   \end{algorithmic}
\end{algorithm}

\subsection{Asymptotic analysis of IGPR}
In this section, we provide some analysis of the IGPR model with a proposal prior and 
in particular we study its asymptotic property with respect to both the cut-off threshold  $\epsilon$ and the sample size $m$. 
Given a likelihood function $\pi(\-d|\theta)$ and a proposal prior $\pi_q(x)$, we can define a joint 
 distribution $$\pi_\epsilon(\theta,\-d) = \pi(\-d|\theta)\pi_q(\theta) \delta_\epsilon(\rho(\-d,\tilde{\-d})),$$
 and as mentioned earlier $\rho(\cdot,\cdot)$ is taken to be the Euclidean distance. 
 Note here that we drop the superscript $j$ in $\theta^j$ without causing any ambiguity.   
Suppose that a training set $D_\epsilon(\tilde{\-d})=\{(\theta_i,\-d_i)\}_{i=1}^m$ is generated according to distribution $\pi_\epsilon(\theta,\-d)$,
and it should be clear that for any two points $\-d_i$ and $\-d_{i'}$ in $D_\epsilon$, we have $\rho(\-d_i,\tilde{\-d})\leq\epsilon$ and 
$\rho(\-d_i,\-d_{i'})\leq2\epsilon$. Next we study the  property of the GP model constructed with $D_{\epsilon}$. 
First we introduce an assumption on the kernel function $k(\-d,\-d')$ of the GP model: 
\begin{assumption}\label{ass:kf}
(a): for any $\-d\in R^{n_d}$, $k(\-d,\-d)=c$ where $c$ is an arbitrary positive constant;
(b): for any $\-d\in R^{n_d}$, and any $\upsilon>0$, there exists $\epsilon>0$, such that, for $\forall$ $\-d'\in R^{n_d}$ satisfying $\rho(\-d,\-d')<\epsilon$,
we have  $|k(\-d,\-d')-k(\-d,\-d)|<\upsilon$.
\end{assumption}
Note that most popular kernel functions such as the squared exponential, exponential,  and the rational quadratic kernels all satisfy this assumption. 
We then have the following theorem: 
\begin{theorem}\label{th:dl}
 For any $\tilde{\-d}\in R^{n_d}$, let $\pi_{GP}(\theta|\tilde{\-d},D_\epsilon(\tilde{\-d}))$ be the GP model constructed with $D_\epsilon(\tilde{\-d})$ and evaluated at $\tilde{\-d}$,
and let $\mu^{m}_\epsilon(\tilde{\-d})$ and $\sigma_\epsilon^{m}(\tilde{\-d})$ be respectively the mean and standard deviation of 
$\pi_{GP}(\theta|\tilde{\-d},D_\epsilon(\tilde{\-d}))$.
Suppose that the kernel function of the GP model  $k(\cdot,\cdot)$ satisfies Assumption \ref{ass:kf}, we have 
\begin{equation}\label{eq:th2}
\mu^{m}_{\epsilon}(\tilde{\-d}) \xrightarrow[\epsilon\rightarrow0]{d}\mu_0^m\quad\mathrm{and}\quad \mu^m_0(\tilde{\-d})\xrightarrow[m\rightarrow\infty]{a.s. }\mathbb{E}_{\btheta|\tilde{\-d}}[\btheta].
\end{equation}

\end{theorem}
\begin{proof} See Appendix~\ref{sec:proof}.
\end{proof}
Loosely speaking, Theorem~\ref{th:dl} states that the mean obtained by the IGPR model converges to the actual posterior mean
as $\epsilon$ approaches to zero and the number of samples increases. 

\section{Numerical examples}
\label{sec:examples}

\subsection{Overview of the examples}\label{sec:overview}
In this section we provide four numerical examples to demonstrate the performance of the proposed IGPR method. 
In the first example, we test the IGRP method on a one-dimensional toy problem, the main purpose of which is to demonstrate some properties of the method
as well as to compare the basic and the adaptive versions of it.
The other three examples concern real-world applications and they are all described by dynamical systems coupling with model errors, which renders 
the likelihood function intractable. 
In these three examples, we compare the performance of IGPR with a representative approximate inference method, the Sequential Neural Likelihood (SNL)~\cite{papamakarios2019sequential}.
SNL is a sampling method which trains an autoregressive flow on simulated data to approximate the likelihood function, and we choose to compare with SNL because a thorough  comparison of SNL with many other popular methods has been conducted in \cite{papamakarios2019sequential}. 

An important issue here is that, in all the three real-world examples the raw data that are recorded are actually time-series signal, which means the dimensionality of the data is extremely high and thus is beyond the capacity of any method for approximate inference. 
A typically remedy for this matter is to construct some summary  statistics of the high dimensional time-series data, and conduct inference using these statistics. 
In these examples we construct the summary statistics following the method in \cite{wang2018adaptive}. 
Specifically the absolute value of the signal amplitude, velocity, acceleration, and the power spectral density (PSD) are extracted from the signal. The mean, variance, skewness, and kurtosis of each physical quantity are computed as the features, resulting in 
a 16-dimensional data  $\bm{d}$  used to infer the model parameters $\btheta$. 
Moreover, for the purpose of performance evaluation, in all the examples the ground truth of these parameters are pre-determined and the data are simulated 
from the underlying model accordingly. Finally, since the wall-clock computational cost is reported in some examples, it is useful to mention that all the experiments are conducted on a desktop computer with a 3.6GHz CPU and 16Gb memory. 

\subsection{A one-dimensional toy problem}

\begin{figure}
  \centering
	\includegraphics[width=1.0\linewidth]{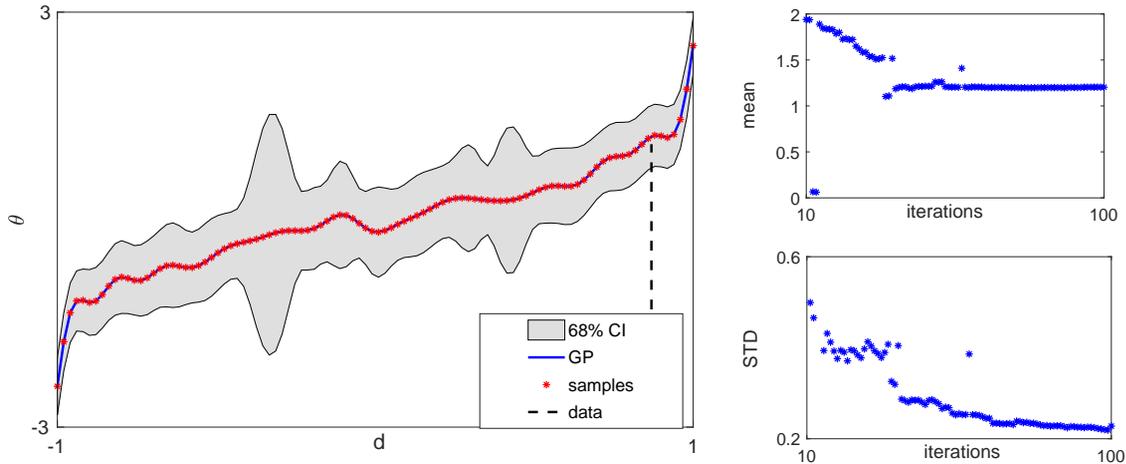}
	\caption{For the 1-D problem. Left: the GP model constructed by Alg.~\ref{alg:igpr1},
	where the solid line is the posterior mean predicted by GP, the shaded area is 
	the $68\%$ confidence interval (CI), and the asterisks are the samples generated from the target likelihood function; right: the top figure is the predicted posterior mean
	and the bottom figure is the predicted posterior standard deviation (STD), both plotted against the number of iterations.} \label{f:eg11}
\end{figure}

\begin{figure}
  \centering
	\includegraphics[width=1.0\linewidth]{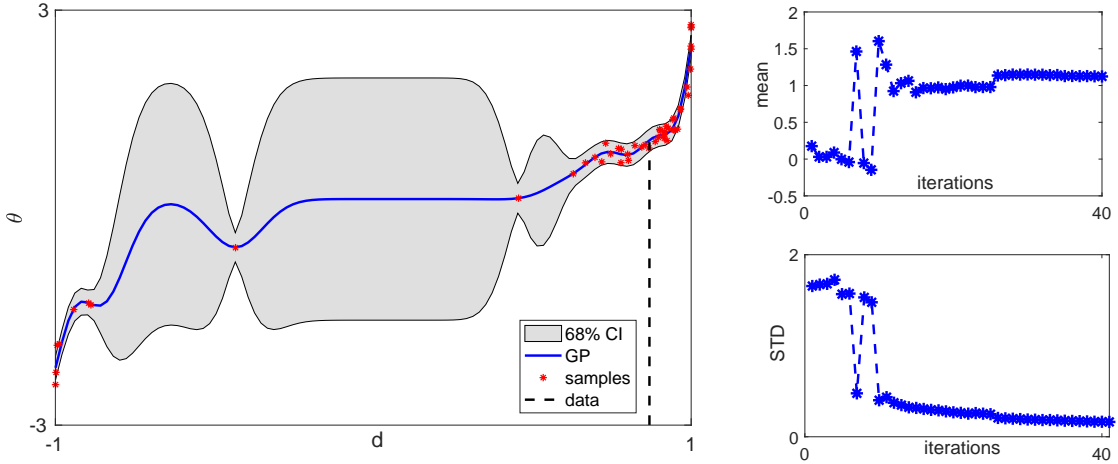}
	\caption{For the 1-D problem. Left: the GP model constructed by Alg.~\ref{alg:igpr},
	where the solid line is the posterior mean predicted by GP, the shaded area is 
	the $68\%$ confidence interval (CI), and the asterisks are the samples generated from the target likelihood function; right: the top figure is the predicted posterior mean
	and the bottom figure is the predicted posterior standard deviation (STD), both plotted against the number of iterations.} \label{f:eg12}
\end{figure}

Our first example is a one-dimensional (1-D) toy problem where the simulation model is
\begin{equation}
d = \mathrm{erf}(\theta+\eta),
\end{equation}
with $\eta\sim\@N(0,0.1^2)$, and $\mathrm{erf}(\cdot)$ being the error function.
The prior distribution is uniform: $\mathrm{U}[-3,3]$, the truth is taken to be $\theta=1$ and the associated observation 
is $d=0.869$. To demonstrate the properties of the  IGPR method, we conduct two sets of simulations here: 
one with the basic version, i.e., Algorithm~\ref{alg:igpr1}, and one with the adaptive version, Algorithm~\ref{alg:igpr}. 

In the test with Algorithm~\ref{alg:igpr1}, the samples are generated directly from the prior distribution.
To evaluate the impact of sample size, we compute the approximate posterior distribution with various sample sizes from 1 to 100. 
In Fig.~\ref{f:eg11} (left), we plot the posterior mean and variance (corresponding to the $68\%$ confidence interval) predicted by the IGPR method
with 100 samples. 
The actual data point $\tilde{d}=0.869$ is indicated by the dashed vertical line in the figure. 
Also shown in the figure are the samples drawn from the likelihood function, which
are largely distributed uniformly in the state space $[-3,3]$ as they are drawn directly from the prior $U[-3,3]$.
 It can be seen that a large portion of the samples actually do not contribute to the GP model as they are rather far from $\tilde{d}$. 
As a result, the variance of the GP model is reduced rather evenly in the state space. 
On the left of Fig.~\ref{f:eg12} we plot the predicted posterior mean and variance 
as a function of the number of iterations, 
in which we observe clear convergence of the mean and variance as the number of iterations increases.

 Next we test the adaptive IGPR method. With this method, we construct the IGPR model with 5 initial points, 
and 40 iterations with one sample in each iteration, resulting in totally 45 samples from the likelihood function. 
Note that in this toy problem we use all the available samples to construct the GP,
which is  slightly different from Alg.~\ref{alg:igpr}.  
In Fig.~\ref{f:eg12} (left), we plot the posterior mean and variance (corresponding to the $68\%$ confidence interval) predicted by the IGPR method
as well as the data points. First we can see here that, unlike the samples in Fig.~\ref{f:eg11} that distributed uniformly, 
most of these samples are placed near the true data point, which demonstrate the 
effectiveness of the adaptive scheme to select samples points. 
As a result, the variance predicted 
by the IGPR model in the area near the data point is much smaller than that in areas that are far apart,
which is also different from Fig.~\ref{f:eg11} where the variance distribution is more even across the whole interval. 
The plots on the right show the mean and variance with respect to the sample size, which shows a much faster convergence 
than those in Fig.~\ref{f:eg11}.  
Finally we report that posterior mean predicted by the adaptive IGPR (Algorithm~\ref{alg:igpr}) is $1.12$ and the standard deviation is $0.16$.

\subsection{Metabolic pathway network}
In terms of real-world applications, we first consider a simply metabolic pathway network, given by a stochastic dynamical system~\cite{voit2004decoupling}
\begin{equation}
    \label{eq:example_S}
    \begin{split}
        \dot{X}_1 &= (\alpha X_2^{-0.4} - \beta_1X_1^{0.5})  \exp(\xi_t),\\
        \dot{X}_2 &= \beta_1X_1^{0.5} - \beta_2X_1^{-1}X_2^{0.4},
    \end{split}
\end{equation}
 where $\xi_t$ is a white Gaussian noise following $\xi_t\sim \mathcal{N}(0, 10^{-2})$, and the initial condition is $X_1(0)=1.2$ and $X_2(0)=1$. In this model, 
$X_1$ and $X_2$ are two dependent metabolites  
and  ($\alpha$, $\beta_1$, $\beta_2$) are three key parameters describing how $X_1$ and $X_2$ are coupled. 
These three parameters can not be measured directly and need to be inferred from the system output. 
In particular we make measurements of the signal $X(t) = X_1(t) + X_2(t)$, and 
as is mentioned earlier, 16 summary statistics  are extracted to infer the  model parameters.
Since this is a synthetic example, the ``true values'' of the parameters are known and shown in Table \ref{tab:S_system_prior}.  Moreover, we assume a uniform prior distribution on each of the parameters where the ranges of the priors are given  in  Table  \ref{tab:S_system_prior}.   
Due to the presence of the model noise $\xi_t$,  the analytical expression of the likelihood  function  in this  problem  is not available, which makes usual posterior estimation methods infeasible.

We implement IGPR and SNL to compute the posterior distributions of the parameters, and 
for the purpose of comparison, we use the same number iterations and the same number of samples per iteration in both methods.
Specifically we fix the number of iterations to be $T=10$, and use three different sample sizes per iteration: $m=50$, $m=100$ and $m=200$,
resulting in 500, $1000$ and $2000$ samples in total respectively. 
In IGPR, we use all the samples for $m=50$ and $100$ since the sample size in each iteration is rather small,  and $\rho=25\%$ for $m=200$.
The tempering scheme is chosen to be $$\sigma_t= 0.1(T-t)/T,$$
which is also used in the next two examples. 
Since the posterior mean is a commonly used estimator for the parameters, we evaluate the performance of the methods by calculating the estimation error between the posterior mean and 
the ground truth for each parameter. 
We repeat the simulation 50 times for either method and compute the average estimation errors 
where the results are summarized in Table~\ref{tab:errormeanAndStd_S}.
The standard deviations (STD) of the  estimation errors are also shown in the table. 
One can see from the table that, with  $500$ samples in total, the results of IGPR are much less accurate than those of SNL,
with 1000 samples it achieves better accuracy  than SNL on $\alpha$ and $\beta_2$,
and finally it produces substantially better results on all three parameters when the sample size becomes 2000.
We also observe from the table that, for sample size 2000, the STDs of the estimation errors are also much smaller
than those of SNL, suggesting that the results of IGPR are subject to smaller variations statistically.  

The performance can also be compared by visualizing the posterior results.
In Fig.~\ref{f:eg2} we show the marginal posterior distributions computed by IGPR and SNL with 2000 total samples in one trial,
where for SNL the distributions are obtained via kernel density estimation of the posterior samples. 
In the figure, we can see that IGPR produces evidently better results, confirming the findings in Table~\ref{tab:errormeanAndStd_S}.
Next in Fig.~\ref{f:eg2_abs}, to demonstrate how the estimation  results improve as the iteration proceeds in IGPR, 
 we plot the average estimation error in IGPR, against
the number of iterations. It can be seen from the figure that, the errors on all three parameters decay obviously as 
the number of iterations increases and it looks that the results become stable after 5 iterations. 

\begin{table}[htbp]
{\footnotesize
  \caption{ For the metabolic pathway example.  The ground truth and the prior distributions of the model parameters.}  \label{tab:S_system_prior}
\begin{center}
  \begin{tabular}{|c|c|c|c|} \hline
     parameter&$\log(\alpha)$& $\log(\beta_1)$& $\log(\beta_2)$\\ \hline
     truth&0 & 0 & 0 \\\hline
  prior& $\mathcal{N}(-0.2,0.2)$ & $\mathcal{N}(-0.2,0.2)$ & $\mathcal{N}(-0.2,0.2)$\\
 \hline
  \end{tabular}
\end{center}
}
\end{table}
\begin{table}[htbp]
{\footnotesize
  \caption{ For the metabolic pathway example. Means and standard deviations of the estimation errors. The smaller errors of the two methods are shown in bold, and
  the numbers inside the parentheses  are the standard deviations.}  \label{tab:errormeanAndStd_S}
\begin{center}
  \begin{tabular}{|c|c|c|c|c|} \hline
    Sample-size &  Method & $\log \alpha$& $\log \beta_1 $& $\log \beta_2$\\ \hline
 500   & IGPR & $0.56(1.796)$ & $0.94(2.665)$ & $2.57( 6.22)$\\
 ($m=50$)  & SNL & ${\bf0.099}(0.053)$ & $ {\bf0.070}(0.047)$ & ${\bf0.079}(0.072)$ \\ \hline
   
    1000   & IGPR & ${\bf0.030}(0.037)$ & ${\bf0.021}(0.030)$ & $0.067(0.117)$ \\
 ($m=100$)  & SNL& $0.079(0.072)$ & $0.080(0.056)$ & $0.067(0.051)$ \\ \hline
   
    2000   & IGPR & ${\bf0.006}(0.007)$ & ${\bf0.004}(0.010)$ & ${\bf0.009}(0.016)$ \\
      ($m=200$)        & SNL& $0.070(0.055)$ & $0.091(0.064)$ & $0.060(0.038)$  \\ \hline
   
  \end{tabular}
\end{center}
}
\end{table}

\begin{figure}
  \centering
	\includegraphics[width=.9\linewidth]{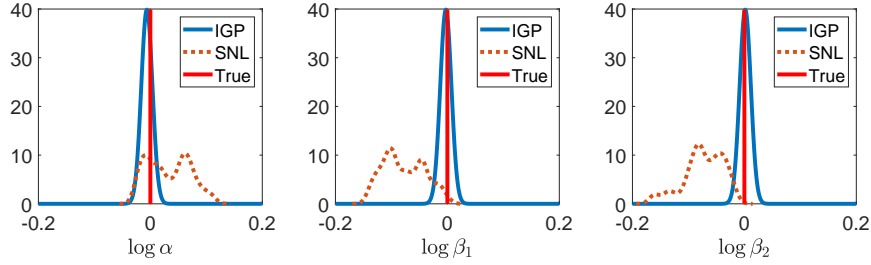}
	\caption{For the metabolic pathway example. The marginal posterior distribution computed by IGPR and SNL.
	In plots, the red solid lines indicate the true parameter values.} 
	\label{f:eg2}
\end{figure}

\begin{figure}
  \centering
	\includegraphics[width=.85\linewidth]{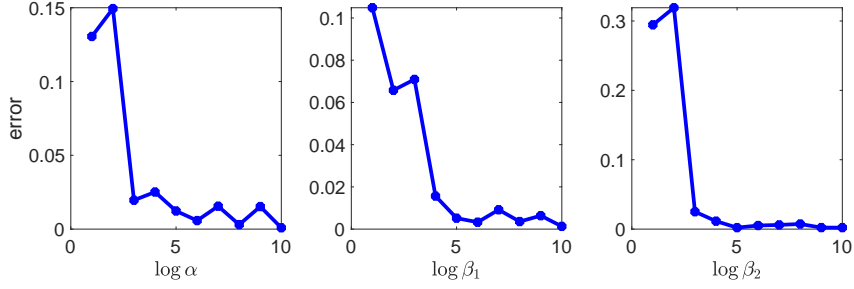}
	\caption{For the metabolic pathway example. The absolute error with respect to iterations computed by IGPR and SNL.} 
	\label{f:eg2_abs}
\end{figure}

\subsection{Blowfly population dynamics}
\label{sec:blowfly}
In this section we consider  a chaotic ecological system tested in \cite{meeds2014gps}. 
It is dynamical system that models the adult blowfly populations~\cite{wood2010statistical}, 
and interestingly the system can produce chaotic behavior for some parameter settings.
Specifically the discretized dynamical system  is the following, 
\begin{equation*}
N_{\tau+1} = PN_{\tau-\tau_\delta}\exp(-N_{t-\tau}/N_0)e_\tau + N_\tau\exp(\partial \epsilon_\tau),\label{e:blowfly}
\end{equation*}
where $N_\tau$ is the population, $e_\tau\sim N(1/\sigma^2_p,1/\sigma^2_p)$ and $\epsilon_\tau \sim N(1/\sigma^2_d,1/\sigma^2_d)$ are sources of noise, and $\tau_\delta$ is an integer. In total there are six model parameters $\tilde{\btheta} = (P,\partial, N_0, \sigma_d, \sigma_p, \tau_\delta )$
that we want to estimate. 
 Like \cite{meeds2014gps} we model the logarithmic value of  $\bm{\theta}$, where the prescribed ground truth and the prior distributions
are both shown in Table~\ref{tab:blowfly_prior}. Time-series signal $\{N_\tau\}$ is generated from model~\eqref{e:blowfly},
and the data for inferring the model parameters are again the 16 summary statistics 
of $N_\tau$ as described in Section~\ref{sec:overview}.

\begin{table}[htbp]
{\footnotesize
  \caption{ For the blowfly population example. The ground truth and the prior distributions of the model parameters.}  \label{tab:blowfly_prior}
\begin{center}
  \begin{tabular}{|c|c|c|c|c|c|c|} \hline
     parameter&$\log P$& $\log \partial $& $\log P_0$& $\log \sigma_d $& $\log \sigma_p $& $\log \tau_\delta$\\ \hline
     truth  & 4 & -1.4 & 6.5 & 0.25 & 0.5 & 2.8\\\hline
  prior&$ N(2, 2^2)$ & ${N}(-1.8, 0.4^2)$ & ${N}(6, 0.5^2)$ &${N}(-0.75, 1^2)$ & ${N}(-0.5, 1^2)$ & ${N}(2.7, 0.1^2)$\\
 \hline
  \end{tabular}
\end{center}
}
\end{table}

\begin{table}[htbp]
{\footnotesize
  \caption{ For the blowfly population example. Means and standard deviations of the estimation errors. The smaller errors of the two methods are shown in bold, and
  the numbers inside the parentheses  are the standard deviations.}  \label{tab:errormeanAndStd_blowfly}
\begin{center}
  \begin{tabular}{|c|c|c|c|c|c|c|c|} \hline
    Sample-size &  Method & $\log P$& $\log \partial $& $\log P_0$& $\log \sigma_d $& $\log \sigma_p $ & $\log \tau_\delta$\\ \hline
 1000   & IGPR & $0.43(0.140)$ & $0.14(0.156)$ & $0.27( 0.102)$ & $ {0.81}(0.671)$ & $0.50(0.192)$ &${\bf0.08}(0.023)$\\
($m=100$)   & SNL & ${\bf0.34}(0.225)$ & $ {\bf0.09}(0.077)$ & $0.27(0.210)$ & ${\bf0.24}(0.173)$ & ${\bf0.44}(0.220)$ & $0.11(0.042)$\\ \hline
   
    5000   & IGPR & ${\bf0.10}(0.066)$ & $0.14(0.021)$ & ${\bf0.09}(0.072)$ & ${\bf0.06}(0.037)$ & $0.21( 0.041)$ & ${\bf0.06}(0.018)$\\
  ($m= 500$)& SNL& $0.20(0.159)$ & $0.14(0.059)$ & $0.11(0.082)$ & $0.12(0.051)$ & ${\bf0.18}( 0.090)$ & $0.07(0.036)$\\ \hline
   
    10000   & IGPR & ${\bf0.05}(0.035)$ & ${\bf0.14}(0.020)$ & ${\bf0.03}(0.027)$ & ${\bf0.05}(0.018)$ & $0.21(0.027)$ & ${\bf0.05}(0.016)$\\
 ($m=1000$)  & SNL& $0.11(0.078)$ & $0.17(0.040)$ & $0.07(0.056)$ & $0.08(0.051)$ & ${\bf0.13}(0.069)$ &$0.06(0.027)$ \\ \hline
   
   $30634$   & SMC-ABC & $0.49$ & $0.37$ & $0.29$ & $0.83$ & $1.02$ & $0.10$\\ \hline
   
  \end{tabular}
\end{center}
}
\end{table}
In the numerical experiments, we also fix the number of iterations to be $T=10$ and test
several different per-iteration sample sizes: $m=100,\,500$ and $1000$, for both IGPR and SNL. 
In IGPR, $\rho$ is chosen so that, for $m=100$, all samples are used to construct the GP model, and for $m=500$ and $1000$, 200 samples are used. 
Once again we test both methods for 50 times and summarize the average estimation errors and their STD in  
Table~\ref{tab:errormeanAndStd_blowfly}. 
First we can see from the table that  the estimation errors 
of both methods are quite small compared to the prior.
Second, similar to the first example, SNL seems to have a better performance for $m=100$, and as the sample size increases,
IGPR becomes more accurate with respect to the estimator error. 
More precisely for $m=1000$ (total sample size 10,000), IGPR yields smaller estimation errors in all but one parameters,
and more importantly, in this case, IGPR results in much smaller STD for all the six parameters, which, once again,
demonstrates that IGPR provides more statistically stable estimates. 
We also provide in the table the estimation errors of a single SMC-ABC simulation with more than 30,000 samples, 
and one can see that the errors are one order of magnitude larger than those of the other two methods;
since SMC-ABC is highly expensive, repeated simulations are not conducted.   

We also plot the marginal posterior distributions obtained by both methods with 10,000 samples
in Fig.~\ref{f:eg3} to provide some visualized comparison, and the figure does show that both methods actually perform well for this example,
with some minor differences on different dimensions. 
In addition to the estimation accuracy, we are also interested the computational cost of the two methods, and thus we show
in Fig.~\ref{f:eg3_timecost} the wall-clock time of the methods, plotted as a function of sample size.  
One can see from the figures that, IGPR is substantially more efficient than SNL and the cost also increases much slower than SNL with respect 
to sample size. For example, for $m=100$, the time cost of SNL is 164 seconds and that
of IGPR is 1.5 seconds, while
for $m=1000$, the cost of SNL and IGPR are 950 and 2.0 seconds respectively.
We note here that the lower computational cost of IGPR is partially due to use of the local GP model which only employs  a small portion of the total samples. 
We should emphasize that an important trade-off for IGPR is that it can only calculate the marginal posterior distributions of the parameters while SNL can obtain the joint one.

\begin{figure}
  \centering
	\includegraphics[width=1\linewidth]{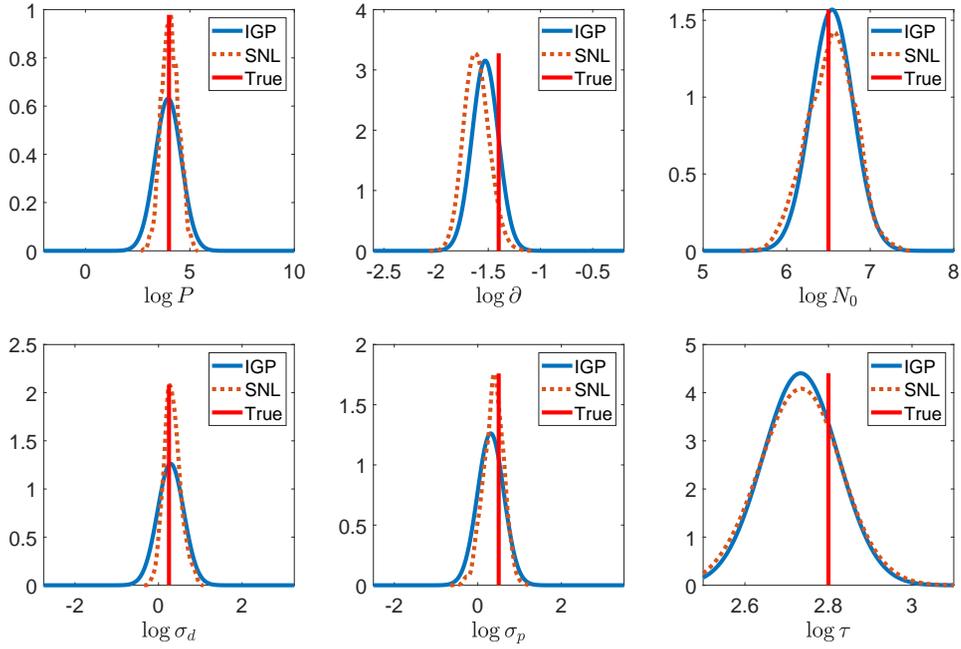}
	\caption{For the blowfly population example. The marginal posterior distribution computed by IGPR and SNL with $10,000$ sample size.
	In both plots, the red solid lines indicate the true parameter values.} 
	\label{f:eg3}
\end{figure}

\begin{figure}
  
  \centerline{\includegraphics[width=.4\linewidth]{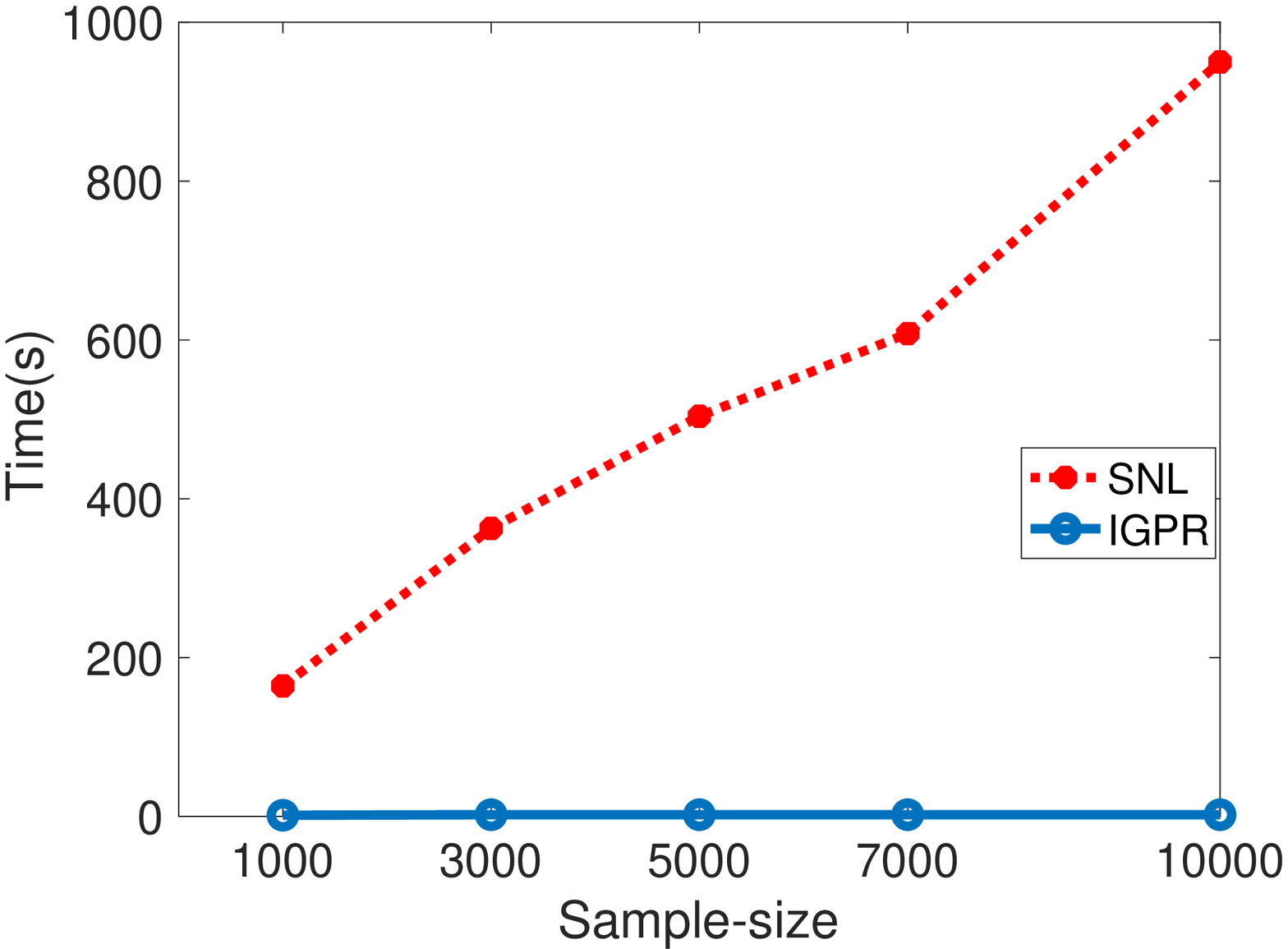}
	\includegraphics[width=.4\linewidth]{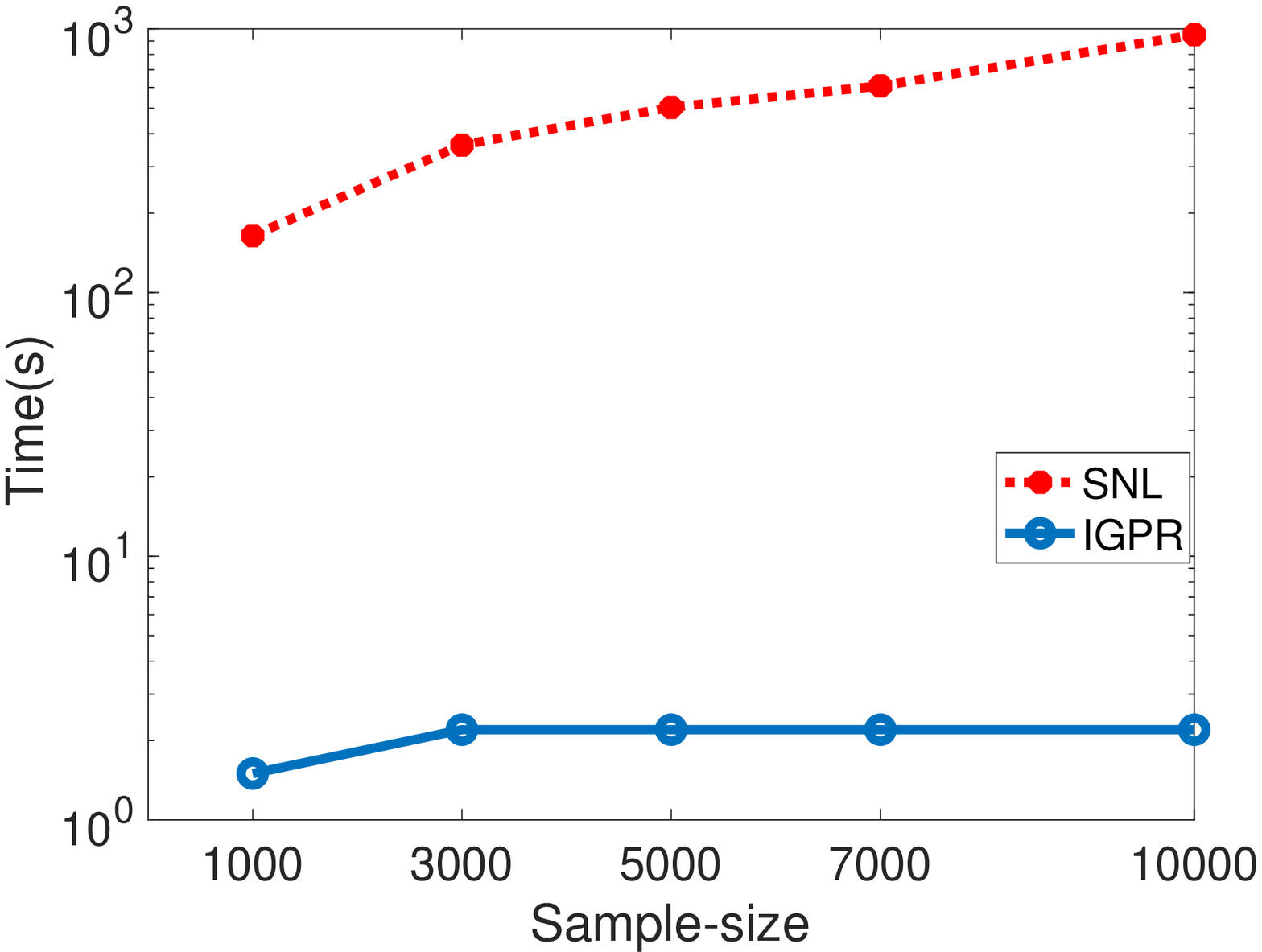}}
	\caption{For the blowfly population example. Time cost (in seconds) of SNL and IGPR with respect to different sample sizes:
	the left figure is on a linear scale and the right one is on a logarithmic scale.} 
	\label{f:eg3_timecost}
\end{figure}

\subsection{Hodgkin-Huxley model}
Our last example is the Hodgkin-Huxley (HH) model that is also used as an application example in \cite{papamakarios2019sequential}.
The HH model describes how the electrical potential measured across the cell membrane of a neuron varies over time as a function of current injected through an electrode. 
 In particular we use the mathematical model \cite{pospischil2008minimal} which consists of  five coupled ordinary differential equations, and is solved numerically using NEURON~\cite{carnevale2006neuron}. 
 The model equations as well as the parameters are detailed in Appendix \ref{sec:hheqn}, and 
our goal is to estimate 12 model parameters~(see Appendix~\ref{sec:hheqn} for further details):
\[    ( g_{\ell},\,\bar{g}_{\mathrm{Na}},\, \bar{g}_{\mathrm{K}}, \bar{g}_{\mathrm{M}}, E_{\ell},\, E_{\mathrm{Na}},\,E_{\mathrm{K}},\,V_{\mathrm{T}},
k_{\beta_{\mathrm{n1} }},\, k_{\beta_{\mathrm{n2}} },\,\tau_{\max },\, \sigma),\]
  from some measured signal.
As is explained in \cite{papamakarios2019sequential}, this problem is both mathematically challenging and of practical interest in neuroscience\cite{daly2015hodgkin,lueckmann2017flexible}.

The data used for inference is the
voltage $V$ recorded for 100 ms, generating a time series of
4001 voltage recordings, and as usual the 16 prescribed statistics are extracted and used as the data $\tilde{\-d}$.  
There are 12 unknown model parameters which are redefined as $\boldsymbol{\theta}=\left(\theta_{1}, \ldots, \theta_{12}\right)$ as is shown 
in Table~\ref{tab:hh_prior}. 
The ground truth of these parameters and the priors used in this example are also provided in Table  \ref{tab:hh_prior}. 


\begin{table}[htbp]
{\footnotesize
  \caption{ For the HH model example. The ground truth and the prior distributions of the model parameters.}  \label{tab:hh_prior}
\begin{center}
  \begin{tabular}{|c|c|c|c|c|c|} \hline
     parameter&$\theta_{1}=\log g_{\ell}$& $\theta_{2}=\log \bar{g}_{\mathrm{Na}}$& $\theta_{3}=\log \bar{g}_{\mathrm{K}}$& $\theta_{4}=\log \bar{g}_{\mathrm{M}} $ \\ \hline
     truth  & -4.60 & 2.99 & 1.60 & -4.96 \\\hline
  prior&$U[-5.30$, $-4.20$] & $U[2.30$, $3.40]$ & $U[0.92$,$2.01]$ & $U[-5.65$, $-4.55]$ \\\hline\hline
    parameter&$\theta_{5}=\log -E_{\ell}$& $\theta_{6}=\log E_{\mathrm{Na}}$& $\theta_{7}=\log -E_{\mathrm{K}}$& $\theta_{8}=\log -V_{\mathrm{T}}$\\ \hline
     truth  & 4.24 & 3.91 & 4.60 & 4.09 \\\hline
  prior& $U[3.55$,$4.65]$& $U[3.21$,$4.32]$& $U[3.91$, $5.01]$ & $U[3.40$, $4.50]$ \\\hline\hline
   parameter& $\theta_{9}=\log k_{\beta_{\mathrm{n1} }}$& $\theta_{10}=\log k_{\beta_{\mathrm{n2}} } $& $\theta_{11}=\log \tau_{\max }$& $ \theta_{12}=\log \sigma$\\ \hline
     truth  & -0.69 & 3.68 & 6.90& 0\\\hline
  prior & $U[-1.39$, $-0.29]$ & $U[3.00$, $4.09]$& $U[6.21$,$7.31]$ & $U[-0.69$,$0.41]$  \\
 \hline
  \end{tabular}
\end{center}
}
\end{table}

\begin{table}[htbp]
{\footnotesize
  \caption{ For the HH model example.  Means and standard deviations of the estimation errors in the HH model. The smaller errors of the two methods are shown in bold, and
  the numbers inside the parentheses  are the standard deviations.}  \label{tab:errormeanAndStd_hh}
\begin{center}
  \begin{tabular}{|c|c|c|c|c|c|} \hline
    Sample-size &  Method & $\log g_\ell$& $\log \bar{g}_{Na} $&  $\log E_{Na}$ & $\log \bar{g}_{\mathrm{M}}$ \\ \hline
 1000   & IGPR & $0.14(0.035)$ & ${\bf0.06}(0.026)$ & ${\bf0.09}( 0.078)$ & $ {0.15}(0.028)$ \\
($m=100$)   & SNL & ${\bf0.12}(0.102)$ & $ 0.10(0.059)$ & $0.27(0.134)$ & ${\bf0.13}(0.065)$ \\ \hline
   
    5000   & IGPR & ${\bf0.12}(0.031)$ & ${\bf0.05}(0.011)$ & ${\bf0.11}(0.068)$ & $0.15(0.015)$ \\
($m=500$)   & SNL             & $0.13(0.110)$ & $0.07(0.031)$ & $0.17(0.093)$ & ${\bf0.11}(0.099)$ \\ \hline
   
    10000   & IGPR & ${\bf0.08}(0.019)$ & $0.04(0.007)$ & ${\bf0.08}(0.039)$ & $0.13(0.009)$ \\
  ($m=1000$) & SNL&                $0.11(0.075)$ & $0.04(0.022)$ & $0.16(0.078)$ & $0.13(0.065)$\\ \hline
   
   Sample-size &  Method & $\log -E_\ell $ & $\log E_{Na}$& $\log -E_K$& $\log -V_T $\\ \hline
 1000   & IGPR & ${\bf0.08}(0.053)$ & ${\bf0.003}(0.002)$ &  $0.18(0.057)$ &${\bf0.03}(0.018)$\\
  ($m=100$) & SNL & $0.15(0.060)$ & $ 0.08(0.073)$ &  ${\bf0.11}(0.076)$ & $0.12(0.099)$\\ \hline
   
    5000   & IGPR & ${\bf0.06}(0.053)$ & ${\bf0.002}(0.001)$ & ${\bf0.15}( 0.028)$ & ${\bf0.03}(0.007)$\\
  ($m=500$) & SNL             & $0.17(0.084)$ & $0.03(0.033)$ & $0.20( 0.058)$ & $0.05(0.035)$\\ \hline
   
    10000   & IGPR & ${\bf0.06}(0.031)$ & ${\bf0.001}(0.0006)$&  ${\bf0.17}(0.018)$ & ${\bf0.03}(0.006)$\\
   ($m=1000$) & SNL        & $0.15(0.092)$ & $0.01(0.013)$&  $0.20(0.042)$ &$0.04(0.020)$ \\ \hline
   
   Sample-size &  Method & $\log k_{\beta_{n1}}$& $\log k_{\beta_{n2}} $& $\log \tau_{max} $ & $\log \sigma$\\ \hline
 1000   & IGPR & ${\bf0.10}( 0.083)$ & $ {\bf0.22}(0.083)$ & ${\bf0.15}(0.042)$ &${\bf0.005}(0.005)$\\
   ($m=100$) & SNL & $0.13(0.077)$ & $0.29(0.077)$ & $0.17(0.137)$ & $0.009(0.006)$\\ \hline
   
    5000   & IGPR & ${\bf0.08}(0.052)$ & $0.22(0.074)$ & ${\bf0.16}( 0.012)$ & ${\bf0.002}(0.002)$\\
    ($m=500$)& SNL             & $0.17(0.087)$ & ${\bf0.20}(0.131)$ & $0.21( 0.124)$ & $0.003(0.003)$\\ \hline
   
    10000   & IGPR  & ${\bf0.08}(0.038)$ & $0.22(0.082)$ & ${\bf0.15}(0.010)$ & ${\bf0.002}(0.002)$\\
   ($m=1000$)& SNL           & $0.16(0.080)$ & ${\bf0.15}(0.128)$ & $0.18(0.100)$ &$0.003(0.003)$ \\ \hline
  \end{tabular}
\end{center}
}
\end{table}
 
The setup of our numerical experiments is  the same as that in Section \ref{sec:blowfly} and is not repeated here.
The average estimation errors of IGPR and SNL are compared in Table~\ref{tab:errormeanAndStd_hh}, and in addition,
 we plot the marginal posteriors with 10,000 samples obtained in one trial in Fig.~\ref{f:eg4}.
From the table we can see that errors in IGPR is either comparable to or evidently smaller than that of SNL,
and in fact the result of $\theta_{10}$ with 10,000 samples is the only case where IGPR yields larger (yet still comparable) error than SNL. 
The marginal distribution plots with 10,000 samples largely show the similar behaviors, 
where one can see that the posterior mean of IGPR is evidently 
closer to the truth for several parameters such as $\theta_1$, $\theta_3$, $\theta_7$ and $\theta_9$.
Moreover  IGPR also result in smaller error STD than SNL in all the cases 
with 5,000 or 10,000 samples, suggesting that IGPR is more statistically stable in this example. 
Finally we compare the computational cost for both methods by plotting their wall-clock time cost in Fig.~\ref{f:eg4_timecost}:
 the plots show that the computational time of SNL grows significantly faster than IGPR
and at sample size 10,000, the time cost for SNL is 1504 seconds while that for IGPR is 12, indicating 
substantial advantage of IGPR in terms of computational efficiency.

\begin{figure}
  \centering
	\includegraphics[width=1\linewidth]{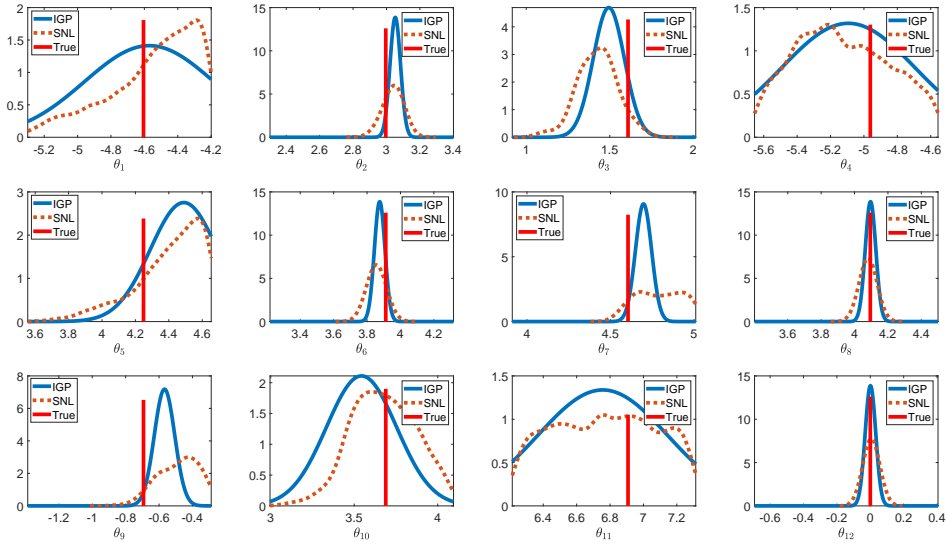}
	\caption{For the HH model example. The marginal posterior distribution computed by IGPR and SNL with $10,000$ sample size.
	In both plots, the red solid lines indicate the true parameter values.} 
	\label{f:eg4}
\end{figure}

\begin{figure}
  
  \centerline{\includegraphics[width=.4\linewidth]{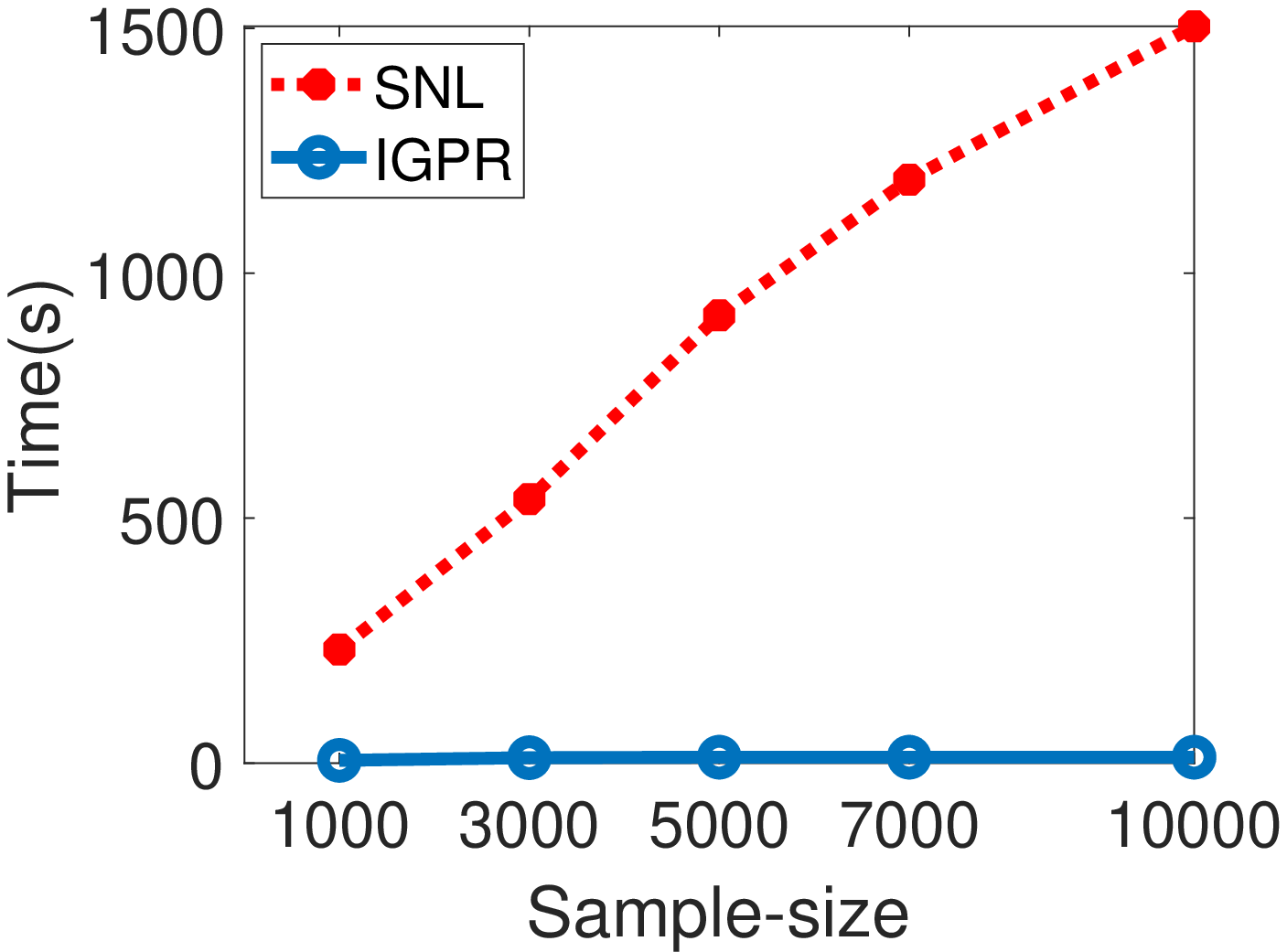}
	\includegraphics[width=.4\linewidth]{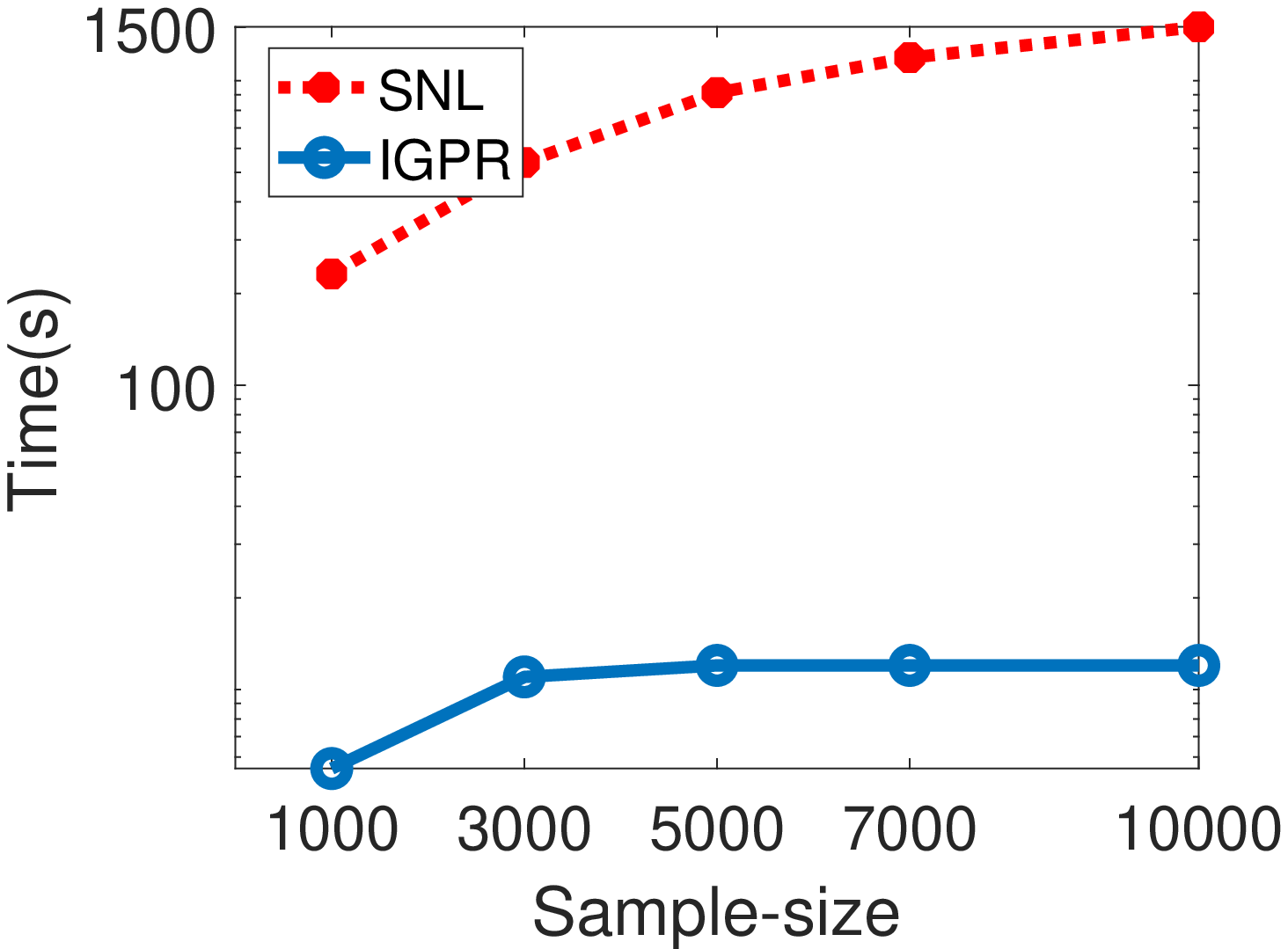}}
	\caption{For the HH model example. Time cost (in seconds) of SNL and IGPR with respect to different sample sizes:
	the left figure is on a linear scale and the right one is on a logarithmic scale.} 
	\label{f:eg4_timecost}
\end{figure}

\section{Conclusions}
\label{sec:conclusions}

In conclusion, we have proposed an efficient method for Bayesian inference problems with
intractable likelihood functions. The method constructs
a GP regression from the output of the underlying simulation model to the input of it, which leads 
to an approximation of the marginal posterior distribution. 
With both mathematical and practical examples, we illustrate that the proposed method can be a competitive tool 
to approximately compute the posterior distributions especially for problems with stringent time constraints. 
Finally we discuss some limitations of the proposed method as well as some issues that need to be address in the future. 
Obviously the approximate posterior distribution obtained by the method is a weighted Gaussian, which thus requires that
the true posterior distribution should not deviate too much from Gaussian. 
Strongly non-Gaussian posteriors such as multimodal distributions may cause problems for the proposed method. 
Second only marginal posterior can be estimated by IGPR method and therefore information lying in the joint distribution such as the 
correlation between parameters can not be obtained. To this end, we believe the multi-output GP methods~\cite{boyle2004dependent} may provide a viable path to 
directly approximating the joint  posterior distribution.  
We plan to address these issues in future studies. 

\appendix '

\section{Proof of Theorem~\ref{th:dl}}\label{sec:proof}

We first prove two  lemmas. The first lemma considers the convergence with respect to $\epsilon$. 
For convenience sake,  for a given training set $D_\epsilon(\tilde{\-d})$,
 we  define $\-X_{\epsilon}=\{\-d_i|(\-d_i,\theta_i)\in D_\epsilon(\tilde{\-d})\}$ and 
 $\-Y_\epsilon =\{\theta_i|(\-d_i,\theta_i)\in D_\epsilon(\tilde{\-d})\}$ to denote the input and output data points in the training set.
\begin{lemma}\label{lemma:ep}
Let $\pi_{GP}(\theta|\tilde{\-d},D_\epsilon(\tilde{\-d}))$ be the GP model constructed with $D_\epsilon(\tilde{\-d})$, evaluated at $\tilde{\-d}$ for $\forall \,\tilde{\-d}\in R^{n_d}$,
and let $\mu^{m}_\epsilon(\tilde{\-d})$ and $\sigma_\epsilon^{n_d}(\tilde{\-d})$ be respectively the mean and standard deviation of 
$\pi_{GP}(\theta|{\tilde{\-d}},D_\epsilon(\tilde{\-d}))$.
Suppose that the kernel function of the GP model  $k(\cdot,\cdot)$ satisfies Assumption \ref{ass:kf},
we have, 
\begin{equation}\label{eq:th_app}
\lim_{\epsilon\rightarrow 0}\mu^{m}_{\epsilon}(\tilde{\-d})\xrightarrow[\epsilon\rightarrow0]{d}\mu^{m}_{0}(\tilde{\-d}).
\end{equation}
\end{lemma}
\begin{proof} 
Using Assumption~\ref{ass:kf} it is easy to show that for any $\upsilon>0$ there exists $\epsilon>0$, such that for any $\-d_i,\,\-d_j\in D_\epsilon(\tilde{\-d})$, 
\begin{equation}
 |k(\tilde{\-d},\-d_i)-c|<\upsilon,\quad\mathrm{and}\quad |k(\-d_i,\-d_j)-c|<\upsilon, 
\end{equation}
which implies a pointwise convergence, 
 \[\lim_{\epsilon\rightarrow 0}K(\tilde{\-d},\-X_\epsilon)(K(\-X_\epsilon,\-X_\epsilon)+\sigma_\zeta^2I)^{-1}=K(\tilde{\-d},\-X_0)(k(\-X_0,\-X_0)+\sigma_\zeta^2I)^{-1},\]
 thanks to the continuity of matrix inversion. 
By definition we have $\-Y_\epsilon \rightarrow \-Y_0$ in distribution, and it then follows that, 
\begin{equation}
\begin{aligned}
\mu^{m}_{\epsilon}(\-d)&=\mu(\tilde{\-d})+K(\tilde{\-d},\-X_\epsilon)(K(\-X_\epsilon,\-X_\epsilon)+\sigma_\zeta^2I)^{-1}(\-Y_\epsilon-\mu({\tilde{\-d}})\-1_m)\\
&\xrightarrow[\epsilon\rightarrow0]{d}\mu(\tilde{\-d})+K(\tilde{\-d},\-X_0)(k(\-X_0,\-X_0)+\sigma_\zeta^2I)^{-1}(\-Y_0-\mu(\tilde{\-d})\-1_m)\\
&=\mu^{m}_{0}(\tilde{\-d}).
\end{aligned}
\end{equation}
\end{proof}

\begin{lemma}\label{lemma:ep0}
Let $\pi_{GP}(\theta|\tilde{\-d},D_0(\tilde{\-d}))$ be the GP model constructed with $D_0(\tilde{\-d})$, evaluated at $\tilde{\-d}$ for $\forall \,\tilde{\-d}\in R^{n_d}$,
and let $\mu^{m}_0(\tilde{\-d})$ and $\sigma_0^{m}(\tilde{\-d})$ be respectively the mean and standard deviation of 
$\pi_{GP}(\theta|\tilde{\-d},D_0(\tilde{\-d}))$.
Suppose that the kernel function of the GP model  $k(\cdot,\cdot)$ satisfies Assumption \ref{ass:kf}, and we have, 
\begin{equation}
\mu^{m}_{0}(\-d)\xrightarrow[m\rightarrow\infty]{a.s.}\mathbb{E}_{\theta|\tilde{\-d}}[\-\theta].
\end{equation}
\end{lemma}
\begin{proof}  As $\epsilon=0$, we have $\-d_i=\tilde{\-d}$ for $\forall \-d_i\in D_0(\tilde{\-d})$, and it follows that $K(\-d,\-X_0)=c\-1_{m}$,
 and $K(\-X_0,\-X_0)=c\-1_{m\times m}$ where $I_{m\times m}$ is a $m\times m$ matrix of ones. From Eq.~\eqref{eq:mu}, we have
 \begin{equation}
 \begin{aligned}
 \mu^{m}_{0}(\tilde{\-d})&=\mu(\tilde{\-d})+K(\tilde{\-d},\-\-X_0)(k(\-X_0,\-X_0)+\nu I)^{-1}(\-Y_0-\mu(\tilde{\-d})\-1_m)\\
 &=\mu(\tilde{\-d})+c\-1^T_{m}(c\-1_{m\times m}+\nu I)^{-1}(\-Y_0-\mu(\tilde{\-d})\-1_m)\\
 &=\mu(\tilde{\-d})+c\-1^T_{m}(c\-1_{m}\-1_{m}^T+\nu I)^{-1}(\-Y_0-\mu(\tilde{\-d})\-1_m)
 \end{aligned}
 \end{equation}
 and applying the Sherman–Morrison formula to the equation above yielding, 
 \begin{equation}
 \begin{aligned}
 \mu^{m}_{0}(\tilde{\-d})&=\mu(\tilde{\-d})+c\-1^T_{m}(\nu^{-1}I-\frac{c\nu^{-2}\-1_{m}\-1_{m}^T}{1+cm\nu^{-1}})(\-Y_0-\mu(\tilde{\-d})\-1_m)\\
 &=\mu(\tilde{\-d})+\frac{c\nu^{-1}\-1^T_{m}}{1+cm\nu^{-1}}(\-Y_0-\mu(\tilde{\-d})\-1_m)\\
 &=\frac{c\nu^{-1}}{1+cm\nu^{-1}}\-1_{m}^T\-Y_0+\frac{1}{1+cm\nu^{-1}}\mu(\tilde{\-d})\\
 &=\frac{cm\nu^{-1}}{1+cm\nu^{-1}}\left(\frac1m\sum_{i=1}^m\theta_i\right)+\frac{1}{1+cm\nu^{-1}}\mu(\tilde{\-d}),
 \end{aligned}
 \end{equation}
 where $\{\theta_i\}_{i=1}^m$ are i.i.d. samples following the distribution $\pi(\theta|\tilde{\-d})$. 
 Thus by the strong law of large numbers we obtain, 
 \begin{equation}
 \mu_0^m(\tilde{\-d})\xrightarrow[m\rightarrow\infty]{a.s.}\mathbb[E]_{\theta|\tilde{\-d}}[\theta].
 \end{equation}
 \end{proof}
Theorem \ref{th:dl} is a direct consequence of Lemmas \ref{lemma:ep} and \ref{lemma:ep0}. 

\section{Description of the HH model}\label{sec:hheqn}
In this section we describe the equations of the HH model. 
The main membrane equation is:
\begin{equation}
	C_m \frac{dV}{dt}=-I_l-I_{Na}-I_K-I_M-I_e,
\end{equation}
where $C_m=1$, and we need to specify the terms $I_l$, $I_{Na}$, $I_K$, $I_M$ and $I_e$ in the equation.
{First  $I_l$ is in the form of (with parameters $g_l, E_l$):}
\begin{equation}
	I_l=g_l(V-E_l).
\end{equation}
 {Next the equations for $I_{Na}$, with parameters $\bar{g}_\mathrm{Na}, E_\mathrm{Na}$ and $V_\mathrm{T}$, are,}
\begin{align*}
I_\mathrm{Na}&=\bar{g}_\mathrm{Na}m^3h(V-E_\mathrm{Na}),\\
\frac{\mathrm{d} m}{\mathrm{d} t}&=\alpha_{m}(V)(1-m)-\beta_{m}(V) m, \\
\frac{\mathrm{d} h}{\mathrm{d} t}&=\alpha_{h}(V)(1-h)-\beta_{h}(V) h, \\
\alpha_{m}&=\frac{-0.32\left(V-V_{T}-13\right)}{\exp \left[-\left(V-V_{T}-13\right) / 4\right]-1}, \\
\beta_{m}&=\frac{0.28\left(V-V_{T}-40\right)}{\exp \left[\left(V-V_{T}-40\right) / 5\right]-1}, \\
\alpha_{h}&=0.128 \exp \left[-\left(V-V_{T}-17\right) / 18\right], \\
\beta_{h}&=\frac{4}{1+\exp \left[-\left(V-V_{T}-40\right) / 5\right]}.
\end{align*}
 The equations for $I_{K}$ depend on four parameters: $\bar{g_\mathrm{K}}, E_\mathrm{K}, \beta_{\mathrm{n1}}, \beta_{\mathrm{n2}}$,
 and are given by,
\begin{align*}
I_\mathrm{K}&=\bar{g_\mathrm{K}}n^4(V-E_\mathrm{K}),\\
\frac{\mathrm{d} n}{\mathrm{d} t}&=\alpha_{n}(V)(1-n)-\beta_{n}(V) n, \\
\alpha_{n}&=\frac{-0.032\left(V-V_{T}-15\right)}{\exp \left[-\left(V-V_{T}-15\right) / 5\right]-1}, \\
\beta_{n}&=\beta_\mathrm{n1} \exp \left[-\left(V-V_{T}-10\right) / \beta_\mathrm{n2}\right].
\end{align*}
The term $I_{M}$ which depends on two parameters $\bar{g}_\mathrm{M}$ and $\tau_\mathrm{max}$, is given as the following, 
\begin{align*}
I_{\mathrm{M}}&=\bar{g}_{\mathrm{M}} p\left(V-E_{\mathrm{K}}\right),\\
\frac{\mathrm{d} p}{\mathrm{d} t}&=\left(p_{\infty}(V)-p\right) / \tau_{p}(V) ,\\
p_{\infty}(V)&=\frac{1}{1+\exp [-(V+35) / 10]}, \\
\tau_{p}(V)&=\frac{\tau_{\max }}{3.3 \exp [(V+35) / 20]+\exp [-(V+35) / 20]}.
\end{align*}
Finally $I_{e}$ is specified as 
\begin{equation*}
	\frac{1}{10}I_{e}\sim N(-0.5\mathrm{nA}, \sigma^2),
\end{equation*}
with a parameter  $\sigma$.
The initial condition is set to be
\begin{equation}
m(0)=h(0)=n(0)=p(0)=0 \quad \text { and } \quad V(0)=-70 \mathrm{mV}.
\end{equation}
The physical meanings of the variables and parameters can be found in \cite{pospischil2008minimal}.

\bibliographystyle{siamplain}
\bibliography{igp}

\end{document}